\documentclass[12pt]{article}
\usepackage{amsmath, amsfonts, amsthm}
\DeclareSymbolFont{bbold}{U}{bbold}{m}{n}
\DeclareSymbolFontAlphabet{\mathbbold}{bbold}

\usepackage{enumerate}
\usepackage{hyperref}
\hypersetup{colorlinks=true, linkcolor=blue, citecolor=blue}
\usepackage[capitalize, nameinlink]{cleveref}
\usepackage{graphicx}
\usepackage{xcolor}

\usepackage{thmtools, thm-restate}
\declaretheorem[name=Lemma]{lemma}
\declaretheorem[name=Definition]{definition}
\declaretheorem[name=Theorem]{theorem}
\declaretheorem[name=Theorem]{priortheorem}

\declaretheorem[name=Corollary]{corollary}

\title{The Quantum Query Complexity of Finding a Tarski Fixed Point on the 2D Grid\footnote{This research was supported by US National Science Foundation grant CCF-2238372.}}
\author{Reed Phillips}

\begin{document}
\maketitle 

\begin{abstract}
    Tarski's theorem states that every monotone function from a complete lattice to itself has a fixed point. We specifically consider the two-dimensional lattice $\mathcal{L}^2_n$ on points $\{1, \ldots, n\}^2$ and where $(x_1, y_1) \leq (x_2, y_2)$ if $x_1 \leq x_2$ and $y_1 \leq y_2$. We show that the quantum query complexity of finding a fixed point given query access to a monotone function on $\mathcal{L}^2_n$ is $\Omega((\log n)^2)$, matching the classical deterministic upper bound of \cite{dang2011computational}. The proof consists of two main parts: a lower bound on the quantum query complexity of a composition of a class of functions including ordered search, and an extremely close relationship between finding Tarski fixed points and nested ordered search.
\end{abstract}

\section{Introduction}

Tarski's fixed-point theorem states that every monotone function from a complete lattice to itself has a fixed point. It has applications in game theory, such as showing the existence of Nash equilibria in supermodular games \cite{etessami2019tarski}. With a careful choice of lattice, it can also formalize the definition of certain recursive functions in denotational semantics.

We consider the grid lattice $\mathcal{L}_n^k = [n]^k$, where we have $x \leq y$ if $x_i \leq y_i$ for all $i \in [k]$. This family of lattices is the most-studied setting for the computational difficulty of finding Tarski fixed points. In the classical black-box model, where the unknown function can only be accessed by querying its values at individual vertices, the best known query complexity upper bound for constant $k$ is $O((\log n)^{\lceil (k+1)/2 \rceil})$ due to \cite{chen2022improved}. The best known lower bounds are $\Omega(k)$ due to \cite{BPR24} and $\Omega(k (\log n)^2/\log k)$ due to \cite{BPR25_arxiv}.

Early study of the problem centered around connections to nested ordered search. The first sub-polynomial upper bound for constant $k$ was the $O((\log n)^k)$-query algorithm of \cite{dang2011computational}, which works by recursively finding fixed points in slices of the grid. For the two-dimensional grid $\mathcal{L}_n^2$ in particular, their algorithm does a series of ``inner" binary searches along one axis to provide the information for an ``outer" binary search along the other axis. This nested ordered search structure was shown to be necessary by \cite{etessami2019tarski}. They constructed a family of \emph{herringbone functions} which hide the unique fixed point on a random path called the \emph{spine}. The most straightforward algorithm for finding such a fixed point is to run binary search on the spine, but finding each spine vertex requires a binary search perpendicular to the spine. They proved an $\Omega((\log n)^2)$ lower bound by showing that any correct algorithm must, with high probability, find $\Omega(\log n)$ spine vertices at a cost of $\Omega(\log n)$ queries each.

We extend this $\Omega((\log n)^2)$ lower bound to the quantum query model by tightening the connection to nested ordered search. Rather than using the nested ordered search structure of herringbone functions to bound what an algorithm must find, we construct a specific family of herringbone functions that exactly embed nested ordered search. We then apply a novel composition theorem to get a quantum lower bound on nested ordered search. While our construction does give a black-box reduction, we choose to use the structure of the spectral adversary method of \cite{BSS01} to ``natively" handle the reduction. This proof strategy is possibly of independent interest, as we use the fact that the adversary matrix is nonnegative. The more general adversary method of \cite{hoyer07negative}, which does allow negative weights, cannot directly be used in this way.

\subsection{Model}

For $k,n \in \mathbb{N}$, let $\mathcal{L}_{n}^k = [n]^k$ be the $k$-dimensional grid of side length $n$.
Let $\leq$ be the binary relation where for vertices ${a} = (a_1, \ldots, a_k) \in \mathcal{L}_n^k$ and ${b} = (b_1, \ldots, b_k) \in \mathcal{L}_n^k$, we have  ${a} \leq {b}$ if and only if $a_i \leq b_i$ for each $i \in [k]$. 
We consider the lattice $(\mathcal{L}_n^k, \leq)$.
A  function $f :\mathcal{L}_n^k \to \mathcal{L}_n^k$ is monotone if ${a} \leq {b}$ implies that $f({a}) \leq f({b})$. 

Let 
 ${TARSKI}(n,k)$ denote the Tarski search problem on the $k$-dimensional grid of side length $n$. 
\begin{definition} [${TARSKI}(n,k)$]
Let $k,n \in \mathbb{N}$.
     Given black-box access to an unknown monotone function $f : \mathcal{L}_n^k \to \mathcal{L}_n^k$, find a vertex $x \in \mathcal{L}_n^k$ with $f(x) = x$ using as few queries as possible.
 \end{definition}

We will work in the quantum query model. While our proofs do not interact directly with the specifics of this model, we provide a formal definition in \cref{sec:preliminaries} for reference. 

An algorithm is said to \emph{succeed} on an input function $f: \mathcal{L}_n^k \to \mathcal{L}_n^k$ if it outputs a fixed point of $f$  with probability at least $2/3$. 
The bounded-error quantum query complexity of $TARSKI(n, k)$ is the minimum number of queries used by a quantum algorithm that succeeds on every input function for  the lattice $\mathcal{L}_n^k$.

\subsection{Our contributions}

Our main result is the following theorem.

\begin{theorem} \label{thm:intro_main}
The  quantum query complexity of ${TARSKI}(n,2)$ is $\Omega((\log n)^2)$.
\end{theorem}

The lower bound of \cref{thm:intro_main} extends to any number of dimensions.

\begin{corollary} \label{cor:k-geq-2-tarski-lower-bound}
 The quantum query complexity of ${TARSKI}(n,k)$ is $\Omega((\log n)^2)$ for all $k \geq 2$.
\end{corollary}

Since there are $O((\log n)^2)$ deterministic upper bounds for the 2D and 3D grids due to \cite{dang2011computational} and \cite{fearnley2022faster}, respectively, \cref{cor:k-geq-2-tarski-lower-bound} implies that the quantum query complexity on the 2D and 3D grids is $\Theta((\log n)^2)$.

\begin{corollary}
    The quantum query complexity of ${TARSKI}(n,k)$ for $k \in \{2, 3\}$ is $\Theta((\log n)^2)$.
\end{corollary}

Our proof is based on a composition theorem for lower bounds generated by the spectral adversary method of \cite{BSS01}. Loosely speaking, we define a \emph{generalized search function} to be a problem for which each query either gives no information about the answer or completely reveals the answer. Letting $SA(f)$ denote the optimal lower bound the spectral adversary method can achieve on a function $f$, we show the following.

\begin{restatable}{theorem}{generalizedSearchComposition}
    \label{thm:composition-with-generalized-search}
    Let $h$ be a composite function $f \circ (g_1, \ldots, g_k)$ as in \cref{def:composite-function}. Suppose that, for all $i \in [k]$, the function $g_i$ is a generalized search function as in \cref{def:generalized-search-function}. Then:
    \begin{align}
        SA(h) &\geq SA(f) \cdot \min_{i \in [k]} SA(g_i)
    \end{align}
\end{restatable}

\section{Related Work}\label{sec:related-work}

Tarski's fixed-point theorem, sometimes called the Knaster-Tarski theorem, was proven in the general form we use by \cite{tarski1955lattice}. It has applications in the study of supermodular games as seen in \cite{etessami2019tarski}.

The problem $TARSKI(n, k)$ was first studied by \cite{dang2011computational}. They showed that, for constant $k \geq 1$, that the query complexity of $TARSKI(n, k)$ is $O((\log n)^k)$. Their algorithm recursively uses binary search to find fixed points of progressively lower-dimensional slices of the grid. A surprising breakthrough came from \cite{fearnley2022faster}, who showed that for constant $k \geq 2$ the query complexity of $TARSKI(n, k)$ is $O\left((\log n)^{\lceil 2k/3 \rceil}\right)$. The core of their algorithm is a subroutine that extracts useful information from a two-dimensional slice in only $O(\log n)$ queries. \cite{chen2022improved} isolated that improvement and used it to construct an $O\left((\log n)^{\lceil (k+1)/2 \rceil}\right)$-query algorithm. \cite{haslebacher26levelset} gave an alternative $O\left((\log n)^{\lceil 2k/3 \rceil}\right)$-query algorithm that uses diagonal slices rather than axis-aligned slices.

The first lower bound beyond $k=1$ was given by \cite{etessami2019tarski}, who formalized a way to embed nested ordered search into $TARSKI(n, k)$. Their \emph{herringbone function} construction gave an $\Omega((\log n)^2)$ randomized lower bound for $TARSKI(n, 2)$ and is the basis for our lower bound. \cite{BPR24} gave the first $k$-dependent lower bounds of $\Omega(k)$ and $\Omega(k \log(n)/\log k)$. \cite{BPR25_arxiv} gave a generalization of herringbone functions to arbitrary $k \geq 2$ and proved a randomized lower bound of $\Omega(k (\log n)^2/\log k)$.

In the white-box model, the position of $TARSKI(n, k)$ within TFNP has been studied. \cite{etessami2019tarski} showed that it is in both PLS and PPAD, which by the results of \cite{fearnley2022cls} implies it is in CLS. The reduction of \cite{CLY23} showed that the special case where the monotone function has a unique fixed point is no easier than the general case.

The spectral adversary method was introduced by \cite{BSS01}. It was subsequently shown by \cite{spalek2006quantumadversaries} to be equivalent to many other methods, in terms of the optimal query complexity lower bounds it could show. A version of the spectral adversary method that allows negative weights was shown to be stronger by \cite{hoyer07negative}, then shown to be optimal up to logarithmic factors by \cite{reichardt09span}.

The spectral adversary method's behavior under function composition has been extensively studied in the Boolean case. \cite{ambainis06polynomial} and \cite{laplante06adversary} showed that composing a function with itself $d$ times raises the resulting adversary bound to the $d$th power. \cite{hoyer2006composition} generalized this result to arbitrary compositions. Their result does not directly apply to our problem because we require non-Boolean function outputs, but their proof does form the basis of ours. \cite{hoyer07negative} partly extended the composition theorem of \cite{hoyer2006composition} to the version with negative weights.

\section{Preliminaries}\label{sec:preliminaries}

\paragraph{Notation.} Let $[n]$ denote the set $\{1, 2, \ldots, n\}$. For a vector $v$, let $\|v\|$ denote its 2-norm. For a matrix $A$, let $\|A\|$ denote its spectral norm. As we use this notation only with non-negative symmetric matrices, this is also the largest eigenvalue of $A$. For two matrices $A$ and $B$, let $A \circ B$ denote the element-wise (or Hadamard) product, which is the matrix with entries 
\begin{align}
    (A \circ B)[x, y] = A[x,y] \cdot B[x,y]\,.
\end{align}
Our proofs in \cref{sec:composition-lower-bounds} feature products of sets of strings, both directly and as the indices of tensor products of matrices. In each case we combine the strings by concatenation, e.g.:
\begin{align}
    \{``ab"\} \times \{``c", ``db"\} = \{``abc", ``abdb"\}
\end{align}

\paragraph{Quantum query model.}
A problem in the quantum query model consists of finite alphabets $G$ and $H$, a domain $S \subseteq G^m$ for some $m \in \mathbb{N}$, and a function $f:S \to H$. The input is a string $s \in S$ and the goal is to compute $f(s)$ by querying characters of $s$. For example, in $TARSKI(n, k)$, we can have $s$ represent a monotone function $\mathcal{L}^k_n \to \mathcal{L}^k_n$ by taking $G=H=[n]^k$ and $m=n^k$. This formulation can only accommodate one fixed point as the output $f(s)$, but we will only use instances with a unique fixed point in our proofs.

The algorithm has a working state $|\phi \rangle$ of the form $|\phi \rangle = \sum_{i, a, z} \alpha_{i, a, z} | i, a, z \rangle $, where $i$ is the label of an index in $[n]$, $a$ is a string representing the answer register, $z$ is a string representing the workplace register, and $\alpha_{i, a, z}$ is a complex amplitude. 
The amplitudes  satisfy the constraint $\sum_{i, a, z} |\alpha_{i, a, z}|^2 = 1$.

Starting from an arbitrary fixed initial state $|\phi_0\rangle$, the algorithm works as an alternating sequence of \emph{query} and \emph{algorithm} steps. If the current state is $| \phi \rangle = \sum_{i,a,z} \alpha_{i,a,z} \left | i, a, z\rangle \right. $, a query transforms it  as follows:
\begin{align} \label{eq:def_U_x}
    \sum_{i,a,z} \alpha_{i,a,z} \left | i, a, z\rangle \right. \rightarrow \sum_{i,a,z} \alpha_{i,a,z} | i, a \oplus s_i, z\rangle,
\end{align}
Where $\oplus$ denotes the bitwise exclusive OR operation, assuming characters of $G$ are represented in binary.
An algorithm step multiplies the vector of $\alpha_{i,a,z}$'s by a unitary matrix that does not depend on $s$.
Thus a $T$-query quantum query algorithm is a sequence of operations
\begin{align}
    U_0 \rightarrow \mathcal{O}_s \rightarrow U_1 \rightarrow \mathcal{O}_s \rightarrow \ldots \rightarrow  U_{T-1} \rightarrow \mathcal{O}_s \rightarrow U_T, 
\end{align}
where $\mathcal{O}_s$ is the oracle gate defined in \eqref{eq:def_U_x} and $U_0, U_1, \ldots, U_T$ are arbitrary  unitary operations that are independent of the input string $s$.

The algorithm is said to \emph{succeed} if at the end it gives a correct answer with probability at least $2/3$, that is:
\begin{align} 
\sum_{i,a,z: i=f(s)} |\alpha_{i,a,z}|^2 \geq 2/3\,.
\end{align}
The bounded-error quantum query complexity is the minimum number of queries used by a quantum algorithm that succeeds on every input string $s \in S$.

\paragraph{Spectral adversary method.} We use the spectral adversary method developed by \cite{BSS01}, in the form stated by \cite{spalek2006quantumadversaries}. It uses the concept of a \emph{distinguisher matrix}, which we will extensively use in our proofs.

\begin{definition}[Distinguisher matrix]
    \label{def:distinguisher-matrix}
    Let $G$ and $H$ be finite alphabets and $m \in \mathbb{N}$. Let $S \subseteq G^m$ and $f : S \to H$ a function. For each $i \in [m]$, the $|S| \times |S|$ \emph{distinguisher matrix} $D^f_i$ (indexed by the elements of $S$) is the matrix with entries:
    \begin{align}
        D^f_i [x, y] = \begin{cases}
            1 & \text{if $x_i \neq y_i$} \\
            0 & \text{if $x_i = y_i$}
        \end{cases}
    \end{align}
\end{definition}

\begin{priortheorem}[Spectral adversary method, see Theorem 3.1 in \cite{spalek2006quantumadversaries}]
\label{thm:spectral-adversary-method}
Let $G$ and $H$ be finite alphabets and $m \in \mathbb{N}$. Let $S \subseteq G^m$ and $f : S \to H$ a function. Let $\Gamma$ denote an $|S| \times |S|$ non-negative symmetric matrix such that for all $x,y \in S$, we have $\Gamma[x, y] = 0$ if $f(x) = f(y)$. Let 
\begin{align}
    SA(f) = \max_{\Gamma} \frac{\|\Gamma\|}{\max_i \|\Gamma \circ D^f_i\|} \,.
\end{align} 
Then for all $\epsilon \in (0,1/2)$, the $\epsilon$-error quantum query complexity of $f$ is at least:
    \begin{align} \label{eq:SA_lb_theorem}
\left(1 - 2 \sqrt{\epsilon(1-\epsilon)}\right) \cdot SA(f) \,.
    \end{align}
\end{priortheorem}

The statement of \cref{thm:spectral-adversary-method} applies generally to $\epsilon$-error quantum query complexity. Our notion of bounded-error quantum query complexity is the $\epsilon=1/3$ case. We will refer to matrices $\Gamma$ satisfying the conditions of \cref{thm:spectral-adversary-method} as \emph{adversary matrices}.

\section{Proof of the Lower Bound}

In a similar way to \cite{etessami2019tarski}, we embed instances of nested ordered search into $TARSKI(n, 2)$. However, the way we connect the difficulty of the embedded instance to the difficulty of $TARSKI(n, 2)$ is different. In particular, the structure of the spectral adversary method allows us to essentially show a reduction from nested ordered search to $TARSKI(n, 2)$.

Our proof will be organized in three sections. In \cref{sec:composition-lower-bounds}, we show our main general result: an extension of the composition lower bound of \cite{hoyer2006composition} to some cases with non-Boolean functions. In \cref{sec:ordered-search}, we apply that result to ordered search to obtain a quantum lower bound for nested ordered search. Finally, in \cref{sec:reduction-to-tarski} we prove the connection between $TARSKI(n, 2)$ and nested ordered search. Complete proofs of the results in these sections are in \cref{app:composition-lower-bounds}, \cref{app:ordered-search-deferred-proofs}, and \cref{app:reduction-deferred-proofs}, respectively.

\subsection{Composition Lower Bounds}
\label{sec:composition-lower-bounds}

We first define the composition of functions in the context of the quantum query model.

\begin{definition}
    \label{def:composite-function}
    Let $\Phi$, $\Sigma$, $\Psi$ be finite alphabets and $k \in \mathbb{N}$. Consider a domain $S_f \subseteq \Sigma^k$ and let $f : S_f \to \Psi$ be a function. For each $i \in [k]$, let $n_i \in \mathbb{N}$. Consider the domains $S_{g_i} \subseteq \Phi^{n_i}$ and let $g_i : S_{g_i} \to \Sigma$ be functions. Let $n = \sum_{i=1}^k n_i$.

    Then define the domain:
    \begin{align}
        S_h &= \bigcup_{x \in S_f} \prod_{i \in [k]} g_i^{-1}(x_i) \subseteq \Phi^n
    \end{align}
    Each $x \in S_h$ can be partitioned into $x^i \in S_{g_i}$ such that $x = x^1\ldots x^k$. Then the composition $h = f \circ (g_1, \ldots, g_k) : S_h \to \Psi$ is the function defined by:
    \begin{align}
        \label{eq:composition-explicit-definition}
        h(x^1\ldots x^k) = f(g_1(x^1)\ldots g_k(x^k))
    \end{align}
    For $x \in S_h$, we will use the notation $\tilde{x} \in S_f$ for the string such that \eqref{eq:composition-explicit-definition} could be written as $h(x) = f(\tilde{x})$.
\end{definition}

Our arguments in this section are similar to \cite{hoyer2006composition}, who investigated $SA(h)$ compared to $SA(f)$ and the $SA(g_i)$ in the case where $\Phi = \Sigma = \Psi = \{0, 1\}$. In particular, the assumption $\Sigma = \{0, 1\}$ was crucial in their arguments. We obtain similar results for general $\Phi$, $\Sigma$, and $\Psi$ under extra assumptions on the inner functions $g_i$. Specifically, we assume them to be \emph{generalized search functions}, which we now define.

\begin{definition}
    \label{def:generalized-search-function}
    Let $\Phi, \Sigma$ be finite alphabets. Let $n \in \mathbb{N}$. Let $S_g \subseteq \Phi^n$ be a domain and consider a function $g : S_g \to \Sigma$. Suppose there exists $m \in \mathbb{N}$ such that $|g^{-1}(\sigma)| = m$ for all $\sigma \in \Sigma$, and label the instances in $S_g$ by unique ordered pairs $(\sigma, j)$ for $\sigma \in \Sigma$ and $j \in [m]$. If this can be done so that:
    \begin{itemize}
        \item For all $(\sigma, j) \in S_g$, we have $g((\sigma, j)) = \sigma$.
        \item For all $i \in [n]$ and all $(\sigma_1, j_1), (\sigma_2, j_2) \in S_g$, whether or not $(\sigma_1, j_1)_i = (\sigma_2, j_2)_i$ depends only on $j_1$, $j_2$, and whether or not $\sigma_1 = \sigma_2$.
    \end{itemize}
    Then we say $g$ is a \emph{generalized search function} with \emph{$m$ variants}, and we will typically refer to instances of such a function by these ordered pairs.
\end{definition}

The intuition for the name comes from how both ordered and unordered search fit this description. More generally, any function where the answer is completely determined by a single (but possibly unknown) query location is a generalized search function. This includes many ``hidden-bit" constructions that have been used to turn search problems into decision problems (see e.g. \cite{Aaronson06}). For example, instead of asking for the location of a fixed point in $TARSKI(n, k)$, one could embed a secret answer $s$ from a finite set $S$ at each fixed point and then ask for $s$.

A natural kind of adversary matrix for generalized search functions is one that treats all of the function's possible outputs symmetrically. We call these matrices \emph{uniform} and formally define them as follows.

\begin{definition}
    \label{def:uniform-adversary-matrix}
    Let $\Phi, \Sigma$ be finite alphabets and let $n, m \in \mathbb{N}$. Let $g$ be a generalized search function with $m$ variants, domain $S_g \subseteq \Phi^n$, and range $\Sigma$. Then an adversary matrix $\Gamma_g \in \mathbb{R}^{|S_g| \times |S_g|}$ is called \emph{uniform} if for some nonnegative symmetric matrix $A \in \mathbb{R}^{m \times m}$, the entries of $\Gamma_g$ satisfy:
    \begin{align}
        \Gamma_g [(\sigma_1, a), (\sigma_2, b)] &= \begin{cases}
            A[a, b] & \text{if $\sigma_1 \neq \sigma_2$} \\
            0 & \text{if $\sigma_1 = \sigma_2$}
        \end{cases}
    \end{align}
    The corresponding $A$ is called the \emph{tile} of $\Gamma_g$.

    The best spectral adversary bound (see \cref{thm:spectral-adversary-method}) achievable by a uniform adversary matrix is denoted $SA^U(g)$.
\end{definition}

While uniform adversary matrices suffice for our lower bound on $TARSKI$, a general composition theorem would be more useful if it covered all adversary matrices. Fortunately, as the following lemma shows, we may consider only uniform adversary matrices without loss of generality. Its proof is based on the automorphism principle of \cite{hoyer07negative} and deferred to \cref{app:composition-lower-bounds}.

\begin{restatable}{lemma}{uniformMatrixSuffices}
\label{lem:uniform-adversary-optimal}
    Let $\Phi, \Sigma$ be finite alphabets and let $n \in \mathbb{N}$. Let $g$ be a generalized search function with $m$ variants, domain $S_g \subseteq \Phi^n$, and range $\Sigma$. Then $SA^U(g) = SA(g)$; that is, there exists a uniform adversary matrix that achieves the optimal bound.
\end{restatable}

We extend the definition of distinguisher matrices to the tiles of uniform adversary matrices as follows.
\begin{definition}[Distinguisher matrix of a tile]
    \label{def:tile-distinguisher-matrix}
    Let $\Phi, \Sigma$ be finite alphabets and let $n \in \mathbb{N}$. Let $g$ be a generalized search function with domain $S_g \subseteq \Phi^n$ and range $\Sigma$. Let $\Gamma_g$ be a uniform adversary matrix for $g$ with tile $A$. Then for each $i \in [n]$, let $D^A_i$ denote the \emph{distinguisher matrix of $A$}: the $|\Sigma| \times |\Sigma|$ matrix (indexed by elements of $\Sigma$) with entries
    \begin{align}
        D^A_i [a, b] = \begin{cases}
            1 & \text{if $(\sigma_1, a)_i \neq (\sigma_2, b)_i$ for all $\sigma_1 \neq \sigma_2 \in \Sigma$} \\
            0 & \text{if $(\sigma_1, a)_i = (\sigma_2, b)_i$ for all $\sigma_1 \neq \sigma_2 \in \Sigma$}
        \end{cases}
    \end{align}
    Which is well-defined because, since $g$ is a generalized search function, whether or not $(\sigma_1, a)_i = (\sigma_2, b)_i$ depends only on $a$, $b$, and the constraint $\sigma_1 \neq \sigma_2$.
\end{definition}

The following lemma shows that we can get all information about a uniform adversary matrix relevant to the spectral adversary method from its tile. Its proof follows from being able to express a uniform adversary matrix as the tensor product of its tile with a simple matrix.

\begin{restatable}{lemma}{uniformAdversaryTile}
    \label{lem:uniform-adversary-bound-in-terms-of-tile}
    Let $\Phi, \Sigma$ be finite alphabets and let $m, n \in \mathbb{N}$. Let $g$ be a generalized search function with $m$ variants, domain $S_g \subseteq \Phi^n$, and range $\Sigma$. Let $\Gamma_g$ be a uniform adversary matrix for $g$ with tile $A \in \mathbb{R}^{m \times m}$. Then:
    \begin{align}
        \min_{i \in [n]} \frac{\|\Gamma_g\|}{\|\Gamma_g \circ D_i^g\|} = \min_{i \in [n]} \frac{\|A\|}{\|A \circ D_i^A\|}
    \end{align}
\end{restatable}

We now give the construction used in our composition theorem.

\begin{definition}[Composition adversary matrix]
    \label{def:composition-adversary-matrix}
    Let $h = f \circ (g_1, \ldots, g_k)$ be a composite function as in \cref{def:composite-function}. Let $\Gamma_f \in \mathbb{R}^{|S_f| \times |S_f|}$ be an adversary matrix for $f$. Suppose that, for each $i \in [k]$, the function $g_i$ is a generalized search function with $m_i$ variants and range $\Sigma$. Let $A_i \in \mathbb{R}^{m_i \times m_i}$ be the tile of a uniform adversary matrix for $g_i$. For each $a, b \in \Sigma$, define the matrix $\Gamma_{g_i}^{(a, b)} \in \mathbb{R}^{m_i \times m_i}$ as:
    \begin{align}
        \Gamma^{(a, b)}_{g_i} &= \begin{cases}
            \|A_i\| I & \text{if $a=b$} \\
            A_i & \text{otherwise}
        \end{cases}
    \end{align}
    Then the \emph{composition adversary matrix} generated by $\Gamma_f$ and the $A_i$ is the matrix $\Gamma_h \in \mathbb{R}^{|S_h| \times |S_h|}$ with entries:
    \begin{align}
        \Gamma_h [x, y] &= \Gamma_f [\tilde{x}, \tilde{y}] \cdot \left(\bigotimes_{i=1}^k \Gamma_{g_i}^{(\tilde{x}_i, \tilde{y}_i)}\right)[x,y]
    \end{align}
\end{definition}

The numerator in the spectral adversary method can then be computed by the following lemma.

\begin{restatable}{lemma}{compositeMatrixNumeratorBound}
    \label{lem:composite-function-matrix-construction}
    Let $h$ be a composite function $f \circ (g_1, \ldots, g_k)$ as in \cref{def:composite-function}. Let $\Gamma_f \in \mathbb{R}^{|S_f| \times |S_f|}$ be a non-negative symmetric matrix. Suppose that, for each $i \in [k]$, the function $g_i$ is a generalized search function with $m_i$ variants. For each $i \in [k]$, let $A_i \in \mathbb{R}^{m_i \times m_i}$ be a non-negative symmetric matrix. 
    Then the composition adversary matrix $\Gamma_h$ of \cref{def:composition-adversary-matrix} generated by $\Gamma_f$ and the $A_i$ satisfies:
    \begin{align}
        \label{eq:gamma-h-norm-equals-gamma-f-times-A}
        \|\Gamma_h\| &= \|\Gamma_f\| \cdot \prod_{i=1}^k \|A_i\|
    \end{align}
\end{restatable}

\begin{proof}[Proof sketch]
    We show \eqref{eq:gamma-h-norm-equals-gamma-f-times-A} by proving inequalities in both directions. To show $\|\Gamma_h\|$ is not too large, we consider an arbitrary unit vector $u \in \mathbb{R}^{S_h}$ and express the product $u^T \Gamma_h u$ in terms of the sub-vectors of $u$ corresponding to each $g_i$. The contributions of sub-vector $i$ are then bounded by $\|A_i\|$ and their combination introduces a factor of $\|\Gamma_f\|$.

    To show that $\|\Gamma_h\|$ is not too small, we explicitly construct an eigenvector and prove it has the correct eigenvalue. It is constructed from principal eigenvectors for $\Gamma_f$ and the $A_i$ analogously to the way $\Gamma_h$ is constructed.
\end{proof}

We will use the following lemma to handle the denominator.

\begin{restatable}{lemma}{compositeMatrixDenominatorBound}
    \label{lem:composite-matrix-respects-D}
    Let $h$ be a composite function $f \circ (g_1, \ldots, g_k)$ as in \cref{def:composite-function} with domain $S_h \subseteq \Phi^n$ for some $n \in \mathbb{N}$ and finite alphabet $\Phi$. Suppose that, for all $j \in [k]$, the function $g_j$ is a generalized search function with range $\Sigma$. 
    Let $\Gamma_f$ be an adversary matrix for $f$. For each $j \in [k]$, let $A_j$ be the tile of some uniform adversary matrix for $g_j$. 
    Let $\Gamma_h$ be the composition adversary matrix generated by $\Gamma_f$ and the $A_j$ as in \cref{def:composition-adversary-matrix}. Let $i \in [n]$, and let $p, q \in \mathbb{N}$ be such that the $i$th character of an input to $h$ is the $q$th character of the corresponding input to $g_p$. Then:
    \begin{align}
        \label{eq:composite-matrix-after-D-norm}
        \|\Gamma_h \circ D_i^h\| &= \|\Gamma_f \circ D^f_p\| \cdot \|A_p \circ D^{A_p}_q\| \cdot \prod_{d \neq p} \|A_d\|
    \end{align}
\end{restatable}

\begin{proof}[Proof sketch]
    Our desired result has a similar form to \cref{lem:composite-function-matrix-construction}. We show that it is in fact exactly the same form by comparing $\Gamma_h \circ D^h_i$ element-by-element to the composite adversary matrix generated by $\Gamma_f \circ D^f_p$, the $A_d$ for $d \neq p$, and $A_p \circ D_p^{A_p}$. This boils down to $D^h_i$ zeroing out exactly the same entries on the left-hand side as $D^f_p$ and $D^{A_p}_q$ do on the right-hand side.
\end{proof}

We now prove our composition theorem.

\generalizedSearchComposition*

\begin{proof}
    Let $\Gamma_f$ be an optimal adversary matrix for $f$ and, for each $i \in [k]$, let $\Gamma_{g_i}$ be an optimal uniform adversary matrix for $g_i$ with tile $A_i$. Let $\Gamma_h$ be the composition adversary matrix generated by $\Gamma_f$ and the $A_i$ as in \cref{def:composition-adversary-matrix}. Consider an arbitrary $j \in [n]$, where $n$ is the length of inputs to $h$, and let $p, q \in \mathbb{N}$ be such that the $j$th character of an input to $h$ is the $q$th character of the part passed to $g_p$. Then by Lemmas \ref{lem:composite-function-matrix-construction} and \ref{lem:composite-matrix-respects-D}, we have:
    \begin{align}
        \frac{\|\Gamma_h\|}{\|\Gamma_h \circ D^h_j\|} &= \frac{\|\Gamma_f\| \cdot \prod_{d=1}^k \|A_d\|}{\|\Gamma_f \circ D_p^f\| \cdot \|A_p \circ D_q^{A_p}\| \cdot \prod_{d \neq p} \|A_d\|} \\
        &= \frac{\|\Gamma_f\|}{\|\Gamma_f \circ D_p^f\|} \cdot \frac{\|A_p\|}{\|A_p \circ D_q^{A_p}\|}
    \end{align}
    Therefore, applying \cref{lem:uniform-adversary-bound-in-terms-of-tile}:
    \begin{align}
        \min_{i \in [n]} \frac{\|\Gamma_h\|}{\|\Gamma_h \circ D^h_i\|} &= \min_{p \in [k]} \left(\frac{\|\Gamma_f\|}{\|\Gamma_f \circ D_p^f\|} \cdot \min_{q \in [m_p]} \frac{\|A_p\|}{\|A_p \circ D_q^{A_p}\|}\right) \\
        &\geq \left(\min_{p \in [k]} \frac{\|\Gamma_f\|}{\|\Gamma_f \circ D_p^f\|}\right) \cdot \left(\min_{i \in [k]} \min_{q \in [m_i]} \frac{\|\Gamma_{g_i}\|}{\|\Gamma_{g_i} \circ D_q^{g_i}\|}\right) \\
        &= SA(f) \cdot \min_{i \in [k]} SA^U(g_i)
    \end{align}
    So using $\Gamma_h$ as an adversary matrix for $h$, we have:
    \begin{align}
        SA(h) \geq SA(f) \cdot \min_{i \in [k]} SA^U(g_i)
    \end{align}
    By \cref{lem:uniform-adversary-optimal} we have $SA^U(g_i) = SA(g_i)$ for all $i \in [k]$, which completes the proof.
\end{proof}

\subsection{Ordered Search}
\label{sec:ordered-search}

In this section, we define the specific version of nested ordered search we will use. We will also apply \cref{thm:composition-with-generalized-search} to construct adversary matrices for it.

\begin{definition}[Ordered search]
    Let $m \in \mathbb{N}$. Then ordered search on $m$ elements, denoted $OS_m$, is the function $OS_m : \{\uparrow^{k-1} * \downarrow^{m-k} \mid 1 \leq k \leq m\} \to [m]$ defined by:
    \begin{align}
        \uparrow^{k-1} * \downarrow^{m-k} \mapsto k
    \end{align}
    That is, given a string with a single $*$ where all characters before the $*$ are $\uparrow$ and all characters after the $*$ are $\downarrow$, return the location of the $*$.
\end{definition}

\begin{definition}[Hidden-symbol ordered search]
    Let $m \in \mathbb{N}$. Then hidden-symbol ordered search on $m$ elements, denoted $HSOS_m$, is the function $HSOS_m : \{\rightarrow^{k-1} x \leftarrow^{m-k} \mid 1 \leq k \leq m, x \in \{\uparrow, \downarrow, *\}\} \to \{\uparrow, \downarrow, *\}$ defined by:
    \begin{align}
        \rightarrow^{k-1} x \leftarrow^{m-k} \mapsto x
    \end{align}
    That is, given a string with a single hidden symbol $x \in \{\uparrow, \downarrow, *\}$ where all characters before the hidden symbol are $\rightarrow$ and all characters after the hidden symbol are $\leftarrow$, return the value of the hidden symbol.
\end{definition}

Note that $HSOS_m$ is a generalized search function as defined in \cref{def:generalized-search-function}, under the labeling where $(\sigma, j)$ refers to $\rightarrow^{j-1} \sigma \leftarrow^{m-j}$.

We will specifically use the composition of hidden-symbol ordered search inside ordered search.

\begin{definition}[Nested ordered search]
    Let $a, b \in \mathbb{N}$. The nested ordered search problem, denoted $NOS_{a, b}$, is the composition $OS_a \circ (HSOS_b, \ldots, HSOS_b)$.
\end{definition}

An example instance of $NOS_{4, 5}$ is the string:
\[
\rightarrow \rightarrow \uparrow \leftarrow \leftarrow,
\rightarrow \rightarrow \rightarrow * \leftarrow,
\rightarrow \downarrow \leftarrow \leftarrow \leftarrow,
\downarrow \leftarrow \leftarrow \leftarrow \leftarrow
\]
For readability, we have inserted commas between the $HSOS_5$ instances. The $*$ is the output of the second instance of $HSOS_5$, so the answer would be $2$.

In order to apply \cref{thm:composition-with-generalized-search} to lower-bound $SA(NOS_{a, b})$, we need to lower-bound $SA(OS_a)$ and $SA(HSOS_b)$. Quantum lower bounds for ordered search already exist, but we need spectral adversary matrices for our two specific variants. Our constructions adapt the methods of \cite{hoyer2002ordered}. Broadly, our adversary matrices use a weight of $1/n$ for instances whose answer (in the case of $OS_a$) or hidden-symbol location (in the case of $HSOS_b$) differ by $n$. The norms of the resulting matrices can be lower-bounded by the harmonic series and upper-bounded by properties of the Hilbert matrix. Their complete proofs are in \cref{app:ordered-search-deferred-proofs}.

\begin{restatable}{lemma}{hsosLowerBound}
    \label{lem:hsos-lower-bound}
    For all $m \in \mathbb{N}$, we have $SA(HSOS_m) \in \Omega(\log m)$.
\end{restatable}
\begin{restatable}{lemma}{osLowerBound}
    \label{lem:os-lower-bound}
    For all $m \in \mathbb{N}$, we have $SA(OS_m) \in \Omega(\log m)$.
\end{restatable}

We can now combine these lower bounds with our composition theorem.

\begin{lemma}
    \label{lem:nos-lower-bound}
    For all $a, b \in \mathbb{N}$, we have $SA(NOS_{a, b}) \in \Omega(\log(a) \log(b))$.
\end{lemma}

\begin{proof}
    Applying \cref{thm:composition-with-generalized-search} with the lower bounds from \cref{lem:hsos-lower-bound} and \cref{lem:os-lower-bound}:
    \begin{align}
        SA(NOS_{a, b}) &\geq SA(OS_a) \cdot SA(HSOS_b) \in \Omega(\log(a) \log(b))
    \end{align}
\end{proof}

\subsection{Reduction to $TARSKI(n, 2)$}
\label{sec:reduction-to-tarski}

We will use a subset of the herringbone functions from \cite{etessami2019tarski}. Each herringbone function has a unique fixed point somewhere on its \emph{spine}, a monotone sequence of connected vertices from the lowest corner of the grid to the highest. On the spine, the function flows towards the fixed point. Off the spine, the function flows diagonally towards the spine. An example can be seen in \cref{fig:herringbone-example}.

\begin{figure}[h]
\begin{center}
    \includegraphics[width=0.4\textwidth]{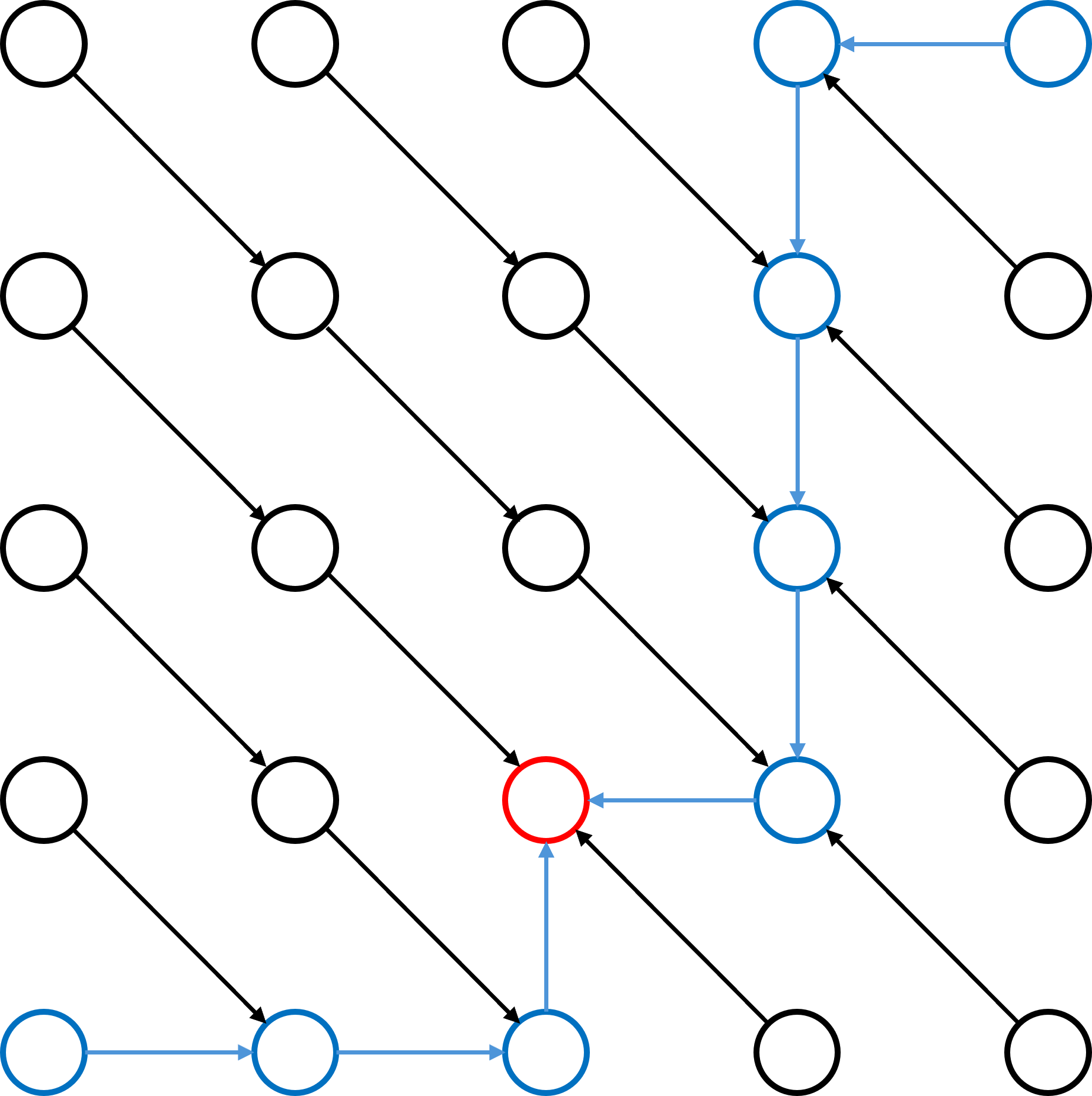}
\end{center}
    \caption{An example herringbone function for $TARSKI(5, 2)$. The blue nodes are the spine, which runs through the red fixed point. Blue arrows flow along the spine while black arrows flow diagonally towards the spine.}
    \label{fig:herringbone-example}
\end{figure}

There is a natural $O((\log n)^2)$ algorithm for finding the fixed point of a herringbone function: run binary search on the spine, where each spine vertex is found by running binary search perpendicular to the spine. This nested binary search algorithm strongly suggests a connection to the nested ordered search problem, and is the basis of the classical lower bound in \cite{etessami2019tarski}. The rest of this section focuses on making that connection tight enough for a direct reduction. For now, we give the formal definition of a herringbone function.

\begin{definition}
    \label{def:herringbone}
    Let $n \in \mathbb{N}$. Let $\vec{s} = \{s^i\}_{i=1}^{2n-1}$ be a monotone connected path in the $n \times n$ grid from $(1, 1)$ to $(n, n)$. Let $j \in [2n-1]$. The herringbone function with spine $\vec{s}$ and fixed point index $j$ is the function $h^{\vec{s}, j} : \mathcal{L}^2_n \to \mathcal{L}^2_n$ defined by:
    \begin{align}
        h^{\vec{s}, j}(v) &= \begin{cases}
            v & \text{if $v = s^j$} \\
            s^{i+1} & \text{if $v = s^i$ for some $i<j$}\\
            s^{i-1} & \text{if $v = s^i$ for some $i>j$}\\
            v + (1, -1) & \text{if $v$ is above the spine}\\
            v + (-1, 1) & \text{if $v$ is below the spine}
        \end{cases}
    \end{align}
    Where by ``above the spine" we mean $v \geq s^i$ for some $s^i$ with the same $x$-coordinate as $v$, and by ``below the spine" we mean $v \leq s^i$ for some $s^i$ with the same $x$-coordinate as $v$.
\end{definition}

We will use herringbone functions with spines of a specific form to prove the connection between $TARSKI(n, 2)$ and nested ordered search. These spines will be made of several straight-line segments as in the following definition.
\begin{definition}
    \label{def:straight-line-segment}
    Let $n \in \mathbb{N}$ and let $u, v \in [n]^2$ be such that $u \leq v$. For $\alpha \in \mathbb{R}$, let $\lfloor \alpha \rceil$ denote the rounding of $\alpha$ to the nearest integer, breaking ties towards the larger result. For $c \in [2n]$, let $L(u, v, c) \in [n]^2$ be the vertex with coordinates:
    \begin{align}
        (L(u, v, c))_1 &= \left\lfloor u_1\frac{v_1 + v_2-c}{v_1 + v_2 - u_1 - u_2} + v_1\frac{c-u_1-u_2}{v_1+v_2-u_1-u_2}\right\rceil \\
        (L(u, v, c))_2 &= c - (L(u, v, c))_1
    \end{align}
    The \emph{grid line} from $u$ to $v$ is the sequence of vertices $L(u, v, c)$ for $c \in \{u_1 + u_2, \ldots, v_1 + v_2\}$.
\end{definition}
We verify that \cref{def:straight-line-segment} gives valid spine segments.
\begin{restatable}{lemma}{straightLineValidSpine}
    \label{lem:straight-line-valid-spine}
    Let $n \in \mathbb{N}$ and let $u, v \in [n]^2$ be such that $u \leq v$. Then the grid line from $u$ to $v$ is a connected monotone path from $u$ to $v$.
\end{restatable}

The intuition behind \cref{def:straight-line-segment} is that grid lines are discrete analogues of Euclidean line segments. Indeed, another way of constructing the grid line from $u$ to $v$ is by rounding points on the line segment between $u$ and $v$ to points in $[n]^2$. The upshot is that moving the endpoints of a grid line continuously moves the points in between. This idea is formally captured by the following two technical lemmas.

Loosely, the first of them states that moving an endpoint cannot cause the rest of the grid line to move in the opposite direction. Its proof follows directly from \cref{def:straight-line-segment}.
\begin{restatable}{lemma}{straightLineMonotone}
    \label{lem:straight-line-monotone}
    Let $n \in \mathbb{N}$. Let $b, c, d \in \mathbb{N}$. Then among $u, v \in [n]^2$ such that $v \geq u$, $u_1+u_2=b$, and $v_1 + v_2 = d$, the $x$-coordinate of $L(u, v, c)$ is:
    \begin{itemize}
        \item Monotone increasing in the $x$-coordinate of $u$ if $c \leq d$, and
        \item Monotone increasing in the $x$-coordinate of $v$ if $c \geq b$.
    \end{itemize}
\end{restatable}

The second one builds on \cref{lem:straight-line-monotone} by showing that, as an endpoint moves, the points on the grid line move continuously in the same direction. It does so by showing that, for any particular point $(x, y) \in [n]^2$, the behavior of a grid line near $(x, y)$ is determined by which of several contiguous regions its endpoints are in.
\begin{restatable}{lemma}{pointSlidingContiguousRegions}
    \label{lem:point-sliding-contiguous-regions}
    Let $n \in \mathbb{N}$. Let $b \leq d \in \mathbb{N}$. Let $S_1, S_2 \subseteq [n]^2$ be such that:
    \begin{itemize}
        \item For all $u \in S_1$, we have $u_1+u_2=b$.
        \item For all $v \in S_2$, we have $v_1+v_2=d$.
        \item For all $u \in S_1$ and $v \in S_2$, we have $u \leq v$.
    \end{itemize}
    Let $(x, y) \in [n]^2$ be such that $b \leq x+y \leq d$. Then for all $u \in S_1$, there exist $\delta_1, \delta_4 \in [n]$ such that:
    \begin{itemize}
        \item For all $v \in S_2$ such that $v_1 \leq \delta_1$, we have $(L(u, v, c))_1 < x$.
        \item For all $v \in S_2$ such that $v_1 \in (\delta_1, \delta_4]$, we have $L(u, v, c)=(x, y)$.
        \item For all $v \in S_2$ such that $v_1 > \delta_4$, we have $(L(u, v, c))_1 > x$.
    \end{itemize}
    And there exist $\delta_2, \delta_3 \in [n]$ such that:
    \begin{itemize}
        \item For all $v \in S_2$ such that $v_1 \in (\delta_1, \delta_2]$, we have $L(u, v, c+1) = L(u, v, c)+(0, 1)$.
        \item For all $v \in S_2$ such that $v_1 \in (\delta_2, \delta_4]$, we have $L(u, v, c+1)=L(u, v, c)+(1, 0)$.
        \item For all $v \in S_2$ such that $v_1 \in (\delta_1, \delta_3]$, we have $L(u, v, c-1) = L(u, v, c)+(-1, 0)$.
        \item For all $v \in S_2$ such that $v_1 \in (\delta_3, \delta_4]$, we have $L(u, v, c-1)=L(u, v, c)+(0, -1)$.
    \end{itemize}
    And the same holds symmetrically: for all $v \in S_2$, there exist $\delta_1, \delta_2, \delta_3, \delta_4$ such that the above hold for all $u \in S_1$.
\end{restatable}

With the grid line segments defined, we now construct a family of spines by concatenating many grid lines together. These spines are constrained to a small \emph{tube} along the main diagonal. The tube itself is further subdivided into \emph{chunks} and those chunks are divided into \emph{regions}. The spine's path through the tube follows grid lines as follows:
\begin{enumerate}[1.]
    \item Start from the lower-left corner $(1, 1)$ and go to a point on the lower-left boundary of the first chunk.
    \item Choose a ``special" region in the first chunk based on the point at which the spine entered. Head directly up-right until entering the special region, then to a point on the upper-right boundary of that region, then directly up-right until entering the second chunk.
    \item Repeat step 2 for all subsequent chunks.
    \item Leave the last chunk and go to the upper-right corner $(n, n)$.
\end{enumerate}

A sketch of such a spine is drawn in \cref{fig:chunked-spine-sketch}. This structure of chunks and regions with a ``special" region in each chunk was used by \cite{etessami2019tarski}. Our key innovation is the correlation between the entry points into each chunk and the choice of special regions.

\begin{figure}[h]
    \includegraphics[width=0.5\textwidth]{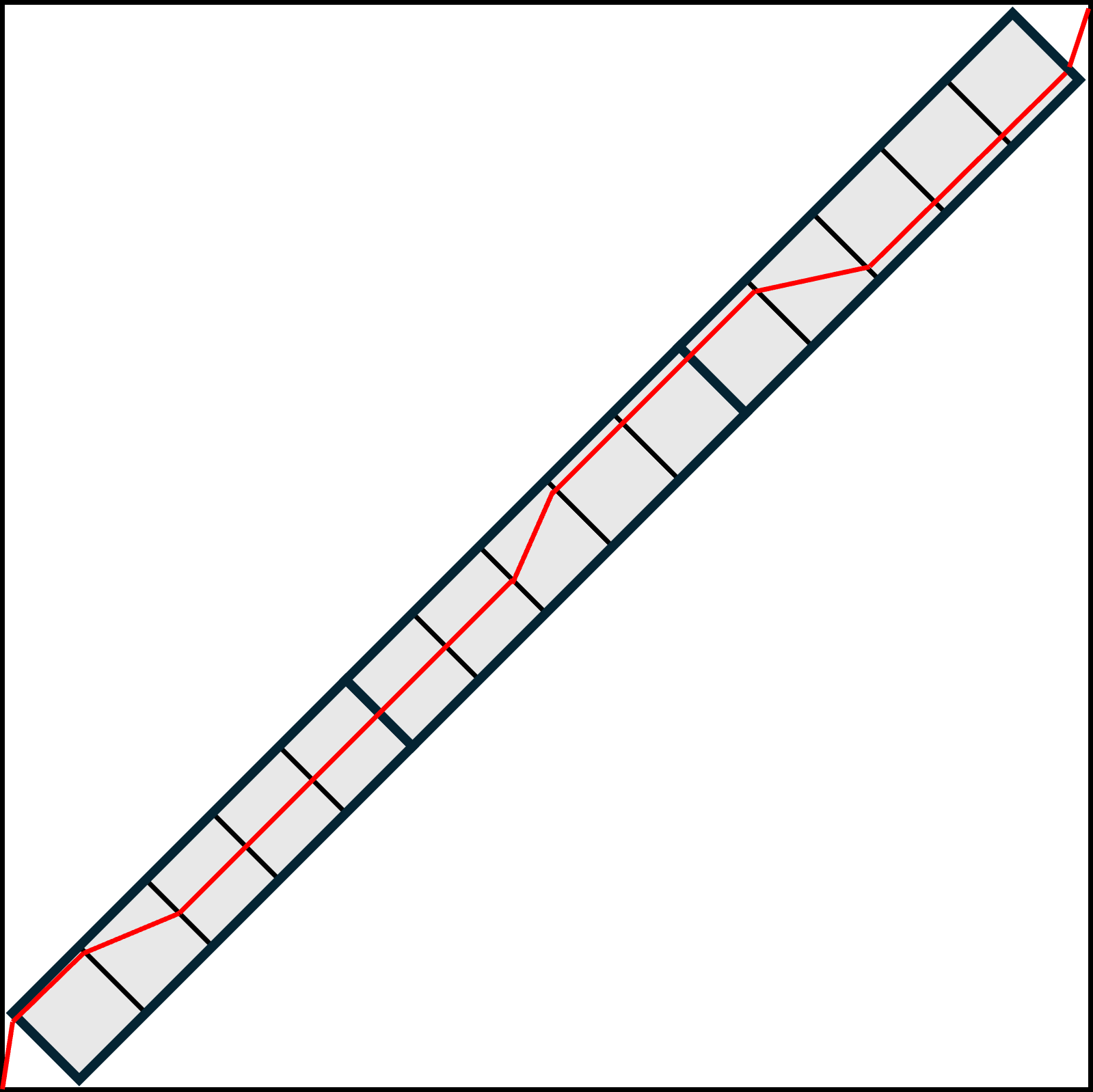}
    \caption{Using the notation of \cref{def:chunked-spine}, an example sketch of a chunked spine for $n=3$. The three chunks are the bold black rectangles, each subdivided into five square regions. The spine is in red. It uses $C = (1, 2, 1, 3)$.}
    \label{fig:chunked-spine-sketch}
\end{figure}

Formally, we first define the chunk and region structure.

\begin{definition}
    \label{def:chunks-regions}
    Let $n \in \mathbb{N}$ and let $n' = n(n^2+n-1)$. For $c \in [2n']$, let $B^c$ denote the set:
    \begin{align}
        B^c &= \{(x, y) \in [n']^2 \mid x+y=c, |x-y| \leq n-1\}
    \end{align}
    And let the elements of $B^c$ be indexed $B^c_1$, $B^c_2$, $\ldots$, ordered by increasing $x$-coordinate. For any two $c \leq d \in [2n']$, let $R(c, d)$ denote the set:
    \begin{align}
        R(c, d) &= \{w \in [n']^2 \mid \exists u \in B^c, v \in B^d \text{s.t. $w$ is in the grid line from $u$ to $v$}\}
    \end{align}
    For $i \in [n]$ and $j \in [n+2]$, define:
    \begin{align}
        \textsc{Low}(i, j) &= 2(n-1)((n+2)(i-1)+j-1)+n+1 \\
        \textsc{High}(i, j) &= 2(n-1)((n+2)(i-1)+j)+n+1
    \end{align}
    Then for $i \in [n]$ and $j \in [n+2]$, we say \emph{region $j$ of chunk $i$} is the set $R(\textsc{Low}(i, j), \textsc{High}(i, j))$.
\end{definition}

We then formally define our \emph{chunked spines}.
\begin{definition}
    \label{def:chunked-spine}
    Let $n \in \mathbb{N}$ and let $n'=n(n^2+n-1)$. Let $C \in [n]^{n+1}$. The \emph{chunked spine following $C$} is constructed by splicing grid lines along each of the following segments:
    \begin{itemize}
        \item Go from $(1, 1)$ to $B^{\textsc{Low}(1, 1)}_{C_1}$.
        \item For each chunk $i \in [n]$:
        \begin{itemize}
            \item Go from $B^{\textsc{Low}(i, 1)}_{C_i}$ to $B^{\textsc{Low}(i, C_i+1)}_{C_i}$.
            \item Go from $B^{\textsc{Low}(i, C_i+1)}_{C_i}$ to $B^{\textsc{High}(i, C_i+1)}_{C_{i+1}}$.
            \item Go from $B^{\textsc{High}(i, C_i+1)}_{C_{i+1}}$ to $B^{\textsc{High}(i, n+2)}_{C_{i+1}}$.
        \end{itemize}
        \item Go from $B^{\textsc{High}(n, n+2)}_{C_{n+1}}$ to $(n', n')$.
    \end{itemize}
\end{definition}

\cref{fig:chunked-spine-sketch} depicts a solid tube of square regions. The following two lemmas show that this picture is accurate. First, that the region boundaries defined by the $B^c$ are of a consistent size and shape.

\begin{restatable}{lemma}{regionBoundarySizeN}
    \label{lem:region-boundary-size-n}
    Let $n \in \mathbb{N}$. For all $i \in \{0, 1, \ldots, n(n+2)\}$, we have:
    \begin{align}
        \left|B^{2(n-1)i+n+1}\right| &= n \label{eq:B-set-size-equals-n}
    \end{align}
    And its elements are, for $j \in [n]$:
    \begin{align}
        B^{2(n-1)i+n+1}_j &= \bigl((n-1)i+j, (n-1)i+n+1-j\bigr) \label{eq:B-elements-characterization}
    \end{align}
\end{restatable}

\begin{proof}[Proof sketch]
    There are three constraints that define points $(x, y) \in B^{2(n-1)i+n+1}$:
    \begin{enumerate}[(i)]
        \item $x+y=2(n-1)i+n+1$,
        \item $|x-y| \leq n$, and
        \item $(x, y) \in [n']^2=[n(n^2+n-1)]^2$.
    \end{enumerate}
    Solving this system of equations implies the lemma.
\end{proof}

Second, that each region actually fills the space between its boundaries. We show this by proving that each region is spanned by the up-right grid lines between corresponding points on its boundaries, which fill the space because the shape of grid lines is translation-invariant.

\begin{restatable}{lemma}{regionShapeRectangular}
    \label{lem:region-shape-rectangular}
    Let $n \in \mathbb{N}$. Let $i \in [n]$ and $j \in [n+2]$. Let $w$ be a vertex in region $j$ of chunk $i$. Then there exists $\ell \in [n]$ such that, for all $c_1, d_1$ such that $\textsc{Low}(c_1, d_1) \leq \textsc{Low}(i, j)$ and all $c_2, d_2$ such that $\textsc{High}(c_2, d_2) \geq \textsc{High}(i, j)$, vertex $w$ is on the grid line from $B^{\textsc{Low}(c_1, d_1)}_{\ell}$ to $B^{\textsc{High}(c_2, d_2)}_{\ell}$.
\end{restatable}

\begin{proof}[Proof sketch]
    The desired $\ell$ is:
    \begin{align}
        \ell &= w_1 - \left\lfloor \frac{w_1 + w_2 - (n+1)}{2} \right\rceil
    \end{align}
    Which we show works by plugging the points from \cref{lem:region-boundary-size-n} into \cref{def:straight-line-segment}.
\end{proof}

We now combine the chunked spines of \cref{def:chunked-spine} with a particular choice of fixed point to get the set of $TARSKI$ instances used in our lower bound.

\begin{definition}
    \label{def:tarski-instances}
    Let $n \in \mathbb{N}$ and let $n' = n(n^2+n-1)$. Let $\textsc{Bound} : [n+1] \to [2n']$ denote the function that maps $i \in [n+1]$ to the $i$th chunk boundary's coordinate sum:
    \begin{align}
        \textsc{Bound}(i) &= \begin{cases}
            \textsc{Low}(i, 1) & \text{if $i \leq n$} \\
            \textsc{High}(n, n+2) & \text{if $i=n+1$}
        \end{cases}
    \end{align}
    Then the set $\mathcal{T}(n')$ of $TARSKI(n', 2)$ instances we consider is:
    \begin{align}
        \mathcal{T}(n') &= \left\{h^{\vec{s}, \textsc{Bound}(i)} \mid \vec{s} \text{ is a chunked spine, } i \in [n+1] \right\}
    \end{align}
\end{definition}

We now prove the lemma that makes up most of our lower bound.

\begin{restatable}{lemma}{tarskiNOSReduction}
    \label{lem:tarski-nos-reduction}
    Let $n \in \mathbb{N}$. Then:
    \begin{align}
        SA(TARSKI(n(n^2+n-1), 2)) \geq \frac{1}{7} SA(NOS_{n+1, n})
    \end{align}
\end{restatable}

\begin{proof}[Proof sketch]
    Let $n' = n(n^2+n-1)$. We will consider only instances of $TARSKI(n', 2)$ in the set $\mathcal{T}(n')$ from \cref{def:tarski-instances}. Each such function is parameterized by a choice of $C \in [n]^{n+1}$ and $i \in [n+1]$. These parameters are also enough to specify an instance of $NOS_{n+1, n}$: the answer locations of the $n+1$ inner $HSOS_n$ instances can be given by $C$, and the answer to the outer $OS_{n+1}$ instance can be given by $i$.

    By matching these parameters, we get a one-to-one correspondence between elements of $\mathcal{T}(n')$ and $NOS_{n+1, n}$. We show that, under this correspondence, each $NOS_{n+1, n}$ instance is exactly embedded into its partner in $\mathcal{T}(n')$. Specifically, querying the $j$th character of the $i$th $HSOS_n$ instance gives the same information as querying $B^{\textsc{Bound}(i)}_j$, the $j$th point on the $i$th chunk boundary. This follows from the spine carrying the information of the outer $OS_{n+1}$ instance through the answers to the $HSOS_n$ instances.

    All that remains is showing that an algorithm cannot get more information by querying other points in $[n']^2$. The points of most concern here are those within the regions, since their values depend on both of the adjacent $HSOS_n$ instances. More worryingly, finding a spine vertex would bypass solving the $HSOS_n$ instances altogether.

    This problem is solved by having $C$ determine both the spine's entry points into each chunk and the region within the chunk that the spine transitions from one to the next. Let $(x, y)$ be an arbitrary point in region $\beta$ of chunk $\alpha$. Let $\gamma$ be the index guaranteed by \cref{lem:region-shape-rectangular} such that that $(x, y)$ is on the grid line from $B^{\textsc{Bound}(\alpha)}_{\gamma}$ to $B^{\textsc{Bound}(\alpha+1)}_{\gamma}$. If two instances $f, g \in \mathcal{T}(n')$ disagree at $(x, y)$ but agree at $B^{\textsc{Bound}(\alpha)}_{\gamma}$ and $B^{\textsc{Bound}(\alpha+1)}_{\gamma}$, it must be that $(x, y)$ lies between the special regions of $f$ and $g$ in chunk $\alpha$. Therefore, if $f$ and $g$ also agree at $B^{\textsc{Bound}(\alpha)}_{\beta-1}$, it must be because region $\beta$ is the special region for both $f$ and $g$, and thus $B^{\textsc{Bound}(\alpha)}_{\beta-1}$ is on both of their spines. We can then use \cref{lem:point-sliding-contiguous-regions} to find four ``threshold" points in $B^{\textsc{Bound}(\alpha+1)}$ that determine where $(x, y)$ is in relation to the spine.

    We therefore have, for any fixed $(x, y)$ in a region, a set $V$ of seven points on chunk boundaries on which two instances $f, g \in \mathcal{T}(n')$ that disagree at $(x, y)$ must not completely agree. This is enough for a black-box reduction, since querying those seven points would give enough information to compute the value at $(x, y)$. Rather than take this approach, our proof uses the structure of the spectral adversary method to ``natively" get a bound from this $V$. Specifically, let $\Gamma_{NOS}$ be an optimal adversary matrix for $NOS_{n+1, n}$. We can then set $\Gamma_{\text{Tarski}}=\Gamma_{NOS}$, following our correspondence between $NOS_{n+1, n}$ and $\mathcal{T}(n')$. Letting $D^{\text{Tarski}}_{v}$ denote the distinguisher matrix for $TARSKI(n', 2)$ when querying $v$, we have:
    \begin{align}
        D^{\text{Tarski}}_{(x, y)} &\leq \sum_{v \in V} D^{\text{Tarski}}_v \, ,
    \end{align}
    Where the inequality is taken elementwise. Because $\Gamma_{\text{Tarski}}$ is nonnegative, we can apply the triangle inequality to get:
    \begin{align}
        \|\Gamma_{\text{Tarski}} \circ D^{\text{Tarski}}_{(x, y)}\| &\leq \sum_{v \in V} \|\Gamma_{\text{Tarski}} \circ D^{\text{Tarski}_v}\| \\
        &\leq 7 \cdot \max_{i \in [n+1], j \in [n]} \|\Gamma_{NOS} \circ D^{NOS}_{i, j}\|\, ,
    \end{align}
    Where the last inequality uses distinguisher matrices $D^{NOS}_{i, j}$ for $NOS_{n+1, n}$ and the fact that $NOS_{n+1, n}$ is exactly embedded into the chunk boundaries. By extending this inequality to all $v\in [n']^2$ through similar arguments, we get:
    \begin{align}
        \max_{v \in [n']^2} \|\Gamma_{\text{Tarski}} \circ D^{\text{Tarski}}_v\| &\leq 7 \cdot \max_{i \in [n+1], j \in [n]} \|\Gamma_{\text{NOS}} \circ D^{NOS}_{i, j}\|
    \end{align}
    Which implies:
    \begin{align}
        SA(TARSKI(n', 2)) &\geq \frac{\|\Gamma_{\text{Tarski}}\|}{\max_{v \in [n']^2} \|\Gamma_{\text{Tarski}} \circ D^{\text{Tarski}}_v\|} \\
        &\geq \frac{1}{7} \cdot \frac{\|\Gamma_{NOS}\|}{\max_{i \in [n+1], j \in [n]} \|\Gamma_{\text{NOS}} \circ D^{NOS}_{i, j}\|} \\
        &= \frac{1}{7} \cdot SA(NOS_{n+1, n})
    \end{align}
    Which completes the proof.
\end{proof}

\cref{lem:tarski-nos-reduction} only applies directly to $TARSKI(n, 2)$ for $n$ of a specific form. For completeness, we use the following lemma to extend it to all $n$ and to extend our lower bound to $TARSKI(n, k)$ for $k \geq 3$. In the classical setting, the monotonicity in $n$ was shown by \cite{BPR24} for the special case of comparing $TARSKI(2, k)$ with $TARSKI(n, k)$. The monotonicity in $k$ was shown by \cite{fearnley2022faster}. We combine their proof ideas into one reduction.

\begin{restatable}{lemma}{tarskiComplexityMonotone}
    \label{lem:complexity-monotone}
    Let $n', n \in \mathbb{N}$ such that $n' \leq n$. Let $k', k \in \mathbb{N}$ such that $k' \leq k$. Then the quantum query complexity of $TARSKI(n, k)$ is at least the quantum query complexity of $TARSKI(n', k')$.
\end{restatable}

We can now prove \cref{thm:intro_main} and its extension \cref{cor:k-geq-2-tarski-lower-bound}.

\begin{proof}[Proof of \cref{thm:intro_main}]
Let $n \in \mathbb{N}$. Let $n' \in \mathbb{N}$ be the largest natural number such that $n' \leq n$ and, for some $a \in \mathbb{N}$, we have $n' = a(a^2+a-1)$. Because $n'$ grows only cubically as a function of $a$, the difference between $n'$ and $n$ is $o(n)$, so $a \in \Theta(n^{1/3})$. By Lemmas \ref{lem:tarski-nos-reduction} and \ref{lem:nos-lower-bound}, we have:
\begin{align}
    SA(TARSKI(n', 2)) &\geq \frac{1}{7} \cdot SA(NOS_{a+1, a}) \\
    &\in \Omega(\log(a+1) \log(a)) \\
    &= \Omega((\log a)^2) \\
    &= \Omega((\log n)^2)
\end{align}
By \cref{thm:spectral-adversary-method}, the quantum query complexity of $TARSKI(n', 2)$ is $\Omega((\log n)^2)$. By \cref{lem:complexity-monotone}, the same lower bound applies to $TARSKI(n, 2)$, which completes the proof.
\end{proof}

\begin{proof}[Proof of \cref{cor:k-geq-2-tarski-lower-bound}]
    By \cref{thm:intro_main}, the query complexity of $TARSKI(n, 2)$ is $\Omega((\log n)^2)$. By \cref{lem:complexity-monotone}, this same bound applies to all $k \geq 2$.
\end{proof}

\section{Conclusion and Future Work}

The one-dimensional $TARSKI(n, 1)$ is effectively the same as ordered search. Combining our lower bound and \cite{dang2011computational}'s upper bound, the same relationship holds between $TARSKI(n, 2)$ and nested ordered search. This trend cannot directly generalize to $k \geq 3$ due to the upper bounds of \cite{fearnley2022faster} and \cite{chen2022improved}, so it would be interesting to see if some other generalization of nested ordered search captures the structure of $TARSKI(n, k)$.

The most direct path forwards for the quantum setting is to extend the classical lower bounds of \cite{BPR24} and \cite{BPR25_arxiv}. They work by constructing a hard distribution of inputs and bound the rate an algorithm can learn information about the input. These arguments could possibly be converted into one of the quantum adversary methods. Notably, the bound of \cite{BPR25_arxiv} uses a generalization of herringbone functions, so our constructions of adversary matrices for nested ordered search may be a useful component.

\bibliographystyle{alpha}

\bibliography{ref}

\appendix

\section{Composition Lower Bound Deferred Proofs}\label{app:composition-lower-bounds}

In this appendix we give complete proofs of the results of \cref{sec:composition-lower-bounds}.

\uniformMatrixSuffices*

\begin{proof}
    The proof is nearly identical to that of the automorphism principle in \cite{hoyer07negative}. Let $\Gamma$ be an optimal adversary matrix for $g$. Assume without loss of generality that it is normalized so that $\max_{i \in [n]} \|\Gamma \circ D^g_i\| = 1$; we then have $\|\Gamma\| = SA(g)$. We construct a uniform adversary matrix $\Gamma'$ that achieves the same bound by averaging various permutations of the rows and columns of $\Gamma$.

    Let $G$ be the group of all permutations of $\Sigma$. For each $\pi \in G$, we have $\pi$ act on $S_g$ by mapping $(\sigma, j)$ to $(\pi(\sigma), j)$ for all $\sigma \in \Sigma$ and $j \in [m]$. We then allow $\pi$ to permute the rows and columns of $\Gamma$ to produce the matrices $\Gamma_{\pi}$ with entries:
    \begin{align}
        \Gamma_{\pi}[x, y] &= \Gamma [\pi(x), \pi(y)]
    \end{align}
    Let $\delta$ be a principal eigenvector of $\Gamma$. We similarly permute its entries to produce vectors $\delta_{\pi}$ with entries:
    \begin{align}
        \delta_{\pi}[x] &= \delta [\pi(x)]
    \end{align}
    Since we used the same permutation of the rows and columns, the vector $\delta_{\pi}$ is a principal eigenvector of $\Gamma_{\pi}$ for all $\pi \in G$. Moreover, since $g$ is a generalized search function, the instances $\pi(x)$ and $\pi(y)$ have the same answer if and only if the instances $x$ and $y$ do. Therefore, all of the $\Gamma_{\pi}$ are valid adversary matrices for $g$.

    Let $\beta \in \mathbb{R}^{|S_g|}$ be the vector with entries:
    \begin{align}
        \beta [x] &= \sqrt{\sum_{\pi \in G} \delta_{\pi}[x]^2}
    \end{align}
    We assume without loss of generality that there does not exist $x \in S_g$ such that $\delta_{\pi}[x] = 0$ for all $\pi \in G$. Removing the rows and columns corresponding to any such $x$ would leave $\|\Gamma\|$ and the $\|\Gamma_{\pi}\|$ unchanged while not increasing the $\|\Gamma \circ D^g_i\|$ or the $\|\Gamma_{\pi} \circ D^g_i\|$, so the resulting adversary bound would be just as good. Under this assumption, all of the entries of $\beta$ are positive. We can therefore define the matrix $B \in \mathbb{R}^{|S_g| \times |S_g|}$ with entries:
    \begin{align}
        B[x, y] &= \frac{1}{\beta[x] \beta[y]}
    \end{align}
    We can now construct our desired uniform adversary matrix $\Gamma'$:
    \begin{align}
        \Gamma' &= B \circ \sum_{\pi \in G} \left(\Gamma_{\pi} \circ \delta_{\pi} \delta_{\pi}^T\right) \label{eq:gamma-prime-matrix-expression}
    \end{align}
    Or equivalently, its entries are:
    \begin{align}
        \Gamma'[x, y] &= \frac{\sum_{\pi \in G} \Gamma [\pi(x), \pi(y)] \delta [\pi(x)] \delta [\pi(y)]}{\beta[x] \beta[y]} \label{eq:gamma-prime-elementwise-expression}
    \end{align}
    We now show that $\Gamma'$ is a uniform adversary matrix. It is an adversary matrix because all of the $\Gamma_{\pi}$ are adversary matrices, so by the definition in \eqref{eq:gamma-prime-matrix-expression} it satisfies $\Gamma'[x, y] = 0$ if $g(x) = g(y)$. It is uniform because both $\Gamma'$ and $\beta$ are produced by summing over all elements of $G$. Specifically, for any $x, y \in S_g$ and $\rho \in G$, we have:
    \begin{align}
        \Gamma'[\rho(x), \rho(y)] &= \frac{\sum_{\pi \in G} \Gamma[\pi(\rho(x)), \pi(\rho(y))] \delta[\pi(\rho(x))]\delta[\pi(\rho(y))]}{\sqrt{\left(\sum_{\pi \in G}\delta_{\pi}[\rho(x)]^2\right)\left(\sum_{\pi \in G}\delta_{\pi}[\rho(y)]^2\right)}} \\
        &= \frac{\sum_{\pi \in G} \Gamma[\pi(x), \pi(y)] \delta[\pi(x)]\delta[\pi(y)]}{\sqrt{\left(\sum_{\pi \in G}\delta_{\pi}[x]^2\right)\left(\sum_{\pi \in G}\delta_{\pi}[y]^2\right)}} \\
        &= \Gamma'[x, y]
    \end{align}
    Since $\Gamma'$ is a uniform adversary matrix, we have:
    \begin{align}
        \frac{\|\Gamma\|}{\max_{i \in [n]} \|\Gamma \circ D^g_i\|} = SA(g) &\geq SA^U(g) \geq \frac{\|\Gamma'\|}{\max_{i \in [n]} \|\Gamma' \circ D^g_i\|} \label{eq:sa-sandwich-with-sau}
    \end{align}
    So if we can show that $\Gamma'$ achieves at least as good a bound as $\Gamma$, we have shown $SA(g)=SA^U(g)$. We first consider the numerator $\|\Gamma'\|$. The vector $\beta/\sqrt{|G|}$ has unit length, since:
    \begin{align}
        \|\beta\| &= \sqrt{\sum_{x \in S_g} \sum_{\pi \in G} \delta[\pi(x)]^2} = \sqrt{\sum_{\pi \in G} \|\delta\|^2} = \sqrt{|G|}
    \end{align}
    So we have:
    \begin{align}
        \|\Gamma'\| &\geq \frac{1}{|G|} \beta^T \Gamma' \beta \\
        &= \frac{1}{|G|} \sum_{x, y \in S_g} \frac{\sum_{\pi \in G} \Gamma[\pi(x), \pi(y)] \delta[\pi(x)] \delta[\pi(y)] \beta[x] \beta[y]}{\beta[x] \beta[y]} \\
        &= \frac{1}{|G|} \sum_{\pi \in G} \delta_{\pi}^T \Gamma_{\pi} \delta_{\pi} \\
        &= \|\Gamma\| \label{eq:gamma-prime-norm-at-least-gamma}
    \end{align}
    We now consider the denominator. Let $i \in [n]$ be arbitrary. By our normalization of $\Gamma$, we have $\|\Gamma \circ D^g_i\| \leq 1$. Therefore, the matrix $I - (\Gamma \circ D^g_i)$ is positive semi-definite. The same holds for all $\Gamma_{\pi}$:
    \begin{align}
        I - (\Gamma_{\pi} \circ D^g_i) \succeq 0 \qquad \forall \pi \in G \label{eq:i-minus-gamma-pi-psd}
    \end{align}
    The matrix $\delta_{\pi} \delta_{\pi}^T$ the outer product of a vector with itself, so it is positive semi-definite. The matrix $B$ can also be written as the outer product of a vector with itself, so it too is positive semi-definite. By the Schur Product Theorem, we have:
    \begin{align}
        (B \circ \delta_{\pi} \delta_{\pi}^T \circ I) - (B \circ \delta_{\pi} \delta_{\pi}^T \circ \Gamma_{\pi} \circ D^g_i) \succeq 0 \qquad \forall \pi \in G \label{eq:delta-delta-i-product-minus-gamma-pi-psd}
    \end{align}
    Summing \eqref{eq:delta-delta-i-product-minus-gamma-pi-psd} over all $\pi \in G$, we have:
    \begin{align}
        \left(I \circ B \circ \sum_{\pi \in G} \delta_{\pi} \delta_{\pi}^T\right) - \left(\left(B \circ \sum_{\pi \in G} \Gamma_{\pi} \circ \delta_{\pi} \delta_{\pi}^T\right) \circ D^g_i\right) \succeq 0 \label{eq:i-b-delta-product-minus-b-gamma-deltas-psd}
    \end{align}
    The subtracted term of \eqref{eq:i-b-delta-product-minus-b-gamma-deltas-psd} is precisely $\Gamma' \circ D^g_i$. We claim the term it is being subtracted from simplifies to $I$ by computing each of its entries:
    \begin{itemize}
        \item The $[x, y]$ entries for $x \neq y$ are all zero since $I[x, y] = 0$ for all such $x, y$.
        \item The $[x, y]$ entries for $x=y$ each simplify to:
        \begin{align}
            \left(I \circ B \circ \sum_{\pi \in G} \delta_{\pi} \delta_{\pi}^T\right)[x, x] &= 1 \cdot \frac{1}{\beta[x]^2} \cdot \sum_{\pi \in G} \delta_{\pi}[x]^2 = 1
        \end{align}
    \end{itemize}
    So \eqref{eq:i-b-delta-product-minus-b-gamma-deltas-psd} simplifies to:
    \begin{align}
        I - (\Gamma' \circ D^g_i) \succeq 0
    \end{align}
    Since $\Gamma' \circ D^g_i$ is non-negative, we therefore have $\|\Gamma' \circ D^g_i\| \leq 1$. By \eqref{eq:gamma-prime-norm-at-least-gamma}, we then have:
    \begin{align}
        \frac{\|\Gamma'\|}{\max_{i \in [n]} \|\Gamma' \circ D^g_i\|} &\geq \frac{\|\Gamma\|}{\max_{i \in [n]} \|\Gamma \circ D^g_i\|}
    \end{align}
    Which ties together the ends of \eqref{eq:sa-sandwich-with-sau}, completing the proof.
\end{proof}

\uniformAdversaryTile*

\begin{proof}
    We will show the stronger claim that the two ratios are equal for all $i \in [n]$. We first consider the numerators $\|\Gamma_g\|$ and $\|A\|$. Let $B \in \mathbb{R}^{|\Sigma| \times |\Sigma|}$ be the matrix with zeroes along the main diagonal and ones everywhere else. As $\Gamma_g$ can be expressed as the tensor product $B \otimes A$, its norm is the product of those norms:
    \begin{align}
        \|\Gamma_g\| &= \|B\| \cdot \|A\|
        \label{eq:gamma-g-norm-equals-b-times-a}
    \end{align}
    For the denominators, consider two cases of the entries of $(\Gamma_g \circ D_i^g)$ for $(\sigma_1, j_1), (\sigma_2, j_2) \in S_g$:
    \begin{itemize}
        \item If $\sigma_1 = \sigma_2$, then $\Gamma_g[(\sigma_1, j_1), (\sigma_2, j_2)] = 0$. Therefore, $D_i^g[(\sigma_1, j_1), (\sigma_2, j_2)]$ does not affect the result.
        \item If $\sigma_1 \neq \sigma_2$, then $D_i^g[(\sigma_1, j_1), (\sigma_2, j_2)] = D_i^A[j_1, j_2]$.
    \end{itemize}
    In both cases, we can replace $D_i^g[(\sigma_1, j_1), (\sigma_2, j_2)]$ with $D_i^A[j_1, j_2]$ without affecting $\Gamma_g \circ D_i^g$. Therefore:
    \begin{align}
        \|\Gamma_g \circ D_i^g\| &= \left\|\left(B \otimes A\right) \circ \left(B \otimes D_i^A\right)\right\| \\
        &= \left\|\left(B \circ B\right) \otimes \left(A \circ D_i^A\right)\right\| \\
        &= \|B\| \cdot \|A \circ D_i^A\|
    \end{align}
    Combined with \eqref{eq:gamma-g-norm-equals-b-times-a}, the factor of $\|B\|$ cancels out, leaving:
    \begin{align}
        \frac{\|\Gamma_g\|}{\|\Gamma_g \circ D_i^g\|} = \frac{\|A\|}{\|A \circ D_i^A\|}
    \end{align}
    Which completes the proof.
\end{proof}

\compositeMatrixNumeratorBound*

\begin{proof}
    We will show the equality \eqref{eq:gamma-h-norm-equals-gamma-f-times-A} by showing inequalities in both directions. We start by showing that $\|\Gamma_h\| \leq \|\Gamma_f\| \cdot \prod_{i=1}^k \|A_i\|$.

    For $a \in S_f$, let $X_a = \{x \in S_h \mid \tilde{x} = a\}$. Then consider an arbitrary unit-length vector $u \in \mathbb{R}^{S_h}$. We can express the product $u^T \Gamma_h u$ as:
    \begin{align}
        \label{eq:ustar-gammah-u-equals-big-nested-sums}
        u^T \Gamma_h u &= \sum_{a, b \in S_f} \Gamma_f[a, b] \cdot \sum_{x \in X_a, y \in X_b} \left(\bigotimes_{i=1}^k \Gamma_{g_i}^{(a_i, b_i)}\right)[x, y] u[x] u[y]
    \end{align}
    For fixed $a$ and $b$, the inner sum of \eqref{eq:ustar-gammah-u-equals-big-nested-sums} can itself be written as a matrix product. Specifically, letting $u_a$ denote the sub-vector of $u$ made up of the components in $X_a$ (and similarly for $u_b$), we have:
    \begin{align}
        \sum_{x \in X_a, y \in X_b} \left(\bigotimes_{i=1}^k \Gamma_{g_i}^{(a_i, b_i)}\right)[x, y] u[x] u[y] &= u_a^T \left(\bigotimes_{i=1}^k \Gamma_{g_i}^{(a_i, b_i)}\right) u_b
    \end{align}
    By the definition of $\Gamma_{g_i}^{(a_i, b_i)}$, we have $\|\Gamma_{g_i}^{(a_i, b_i)}\| = \|A_i\|$ for all $i$. Therefore, we can upper-bound the product by:
    \begin{align}
        u_a^T \left(\bigotimes_{i=1}^k \Gamma_{g_i}^{(a_i, b_i)}\right) u_b &\leq \|u_a\| \|u_b\| \prod_{i=1}^k \|A_i\|
    \end{align}
    Substituting back into \eqref{eq:ustar-gammah-u-equals-big-nested-sums} and using the fact that $\Gamma_f$ is non-negative, we get:
    \begin{align}
        u^T \Gamma_h u &\leq \sum_{a, b \in S_f} \Gamma_f[a, b] \|u_a\| \|u_b\| \prod_{i=1}^k \|A_i\| \\
        &= \left(\prod_{i=1}^k \|A_i\|\right) \cdot \sum_{a, b \in S_f} \Gamma_f[a, b] \|u_a\| \|u_b\|
        \label{eq:ustar-gammah-u-less-than-sum-of-f}
    \end{align}
    The sum in \eqref{eq:ustar-gammah-u-less-than-sum-of-f} is itself a matrix product: specifically, the vector with components $\|u_a\|$ for all $a \in S_f$, both left- and right-multiplied by $\Gamma_f$. Therefore:
    \begin{align}
        u^T \Gamma_h u &\leq \left(\prod_{i=1}^k \|A_i\|\right) \cdot \left(\sqrt{\sum_{a \in S_f} \|u_a\|^2}\right)^2 \|\Gamma_f\| \\
        &= \left(\prod_{i=1}^k \|A_i\|\right) \cdot \|u\|^2 \|\Gamma_f\| \\
        &= \|\Gamma_f\| \cdot \left(\prod_{i=1}^k \|A_i\|\right)
        \label{eq:ustar-gammah-u-less-than-final-result}
    \end{align}
    As \eqref{eq:ustar-gammah-u-less-than-final-result} holds for all unit-length vectors $u$, we have $\|\Gamma_h\| \leq \|\Gamma_f\| \cdot \prod_{i=1}^k \|A_i\|$ as desired.

    We now show the other direction: $\|\Gamma_h\| \geq \|\Gamma_f\| \cdot \prod_{i=1}^k \|A_i\|$. We do this by constructing an eigenvector with that eigenvalue. Let $\delta_f$ be a principal unit-length eigenvector of $\Gamma_f$. For each $i \in [k]$, let $\delta_{A_i}$ be a principal unit-length eigenvector of $A_i$. Then define $\delta_h \in \mathbb{R}^{S_h}$ to be the vector with components:
    \begin{align}
        \delta_h[x] &= \delta_f[\tilde{x}] \cdot \left(\bigotimes_{i=1}^k \delta_{A_i}\right)[x]
    \end{align}
    Note that because both $\delta_f$ and all of the $\delta_{A_i}$ have unit length, so does $\delta_h$.

    To work through the algebra, we will need the following claim:
    \begin{align}
        \label{eq:delta-ai-product-constant}
        \delta_{A_i}^T \Gamma_{g_i}^{(a, b)} \delta_{A_i} = \|A_i\| \quad \forall i \in [k] \text{ and } a, b \in \Sigma
    \end{align}
    We break this claim into two cases.
    \begin{itemize}
        \item Case $a=b$. Then $\Gamma_{g_i}^{(a, b)} = \|A_i\|I$, and because $\delta_{A_i}$ is a unit vector the claim immediately follows.
        \item Case $a \neq b$. Then $\Gamma_{g_i}^{(a, b)} = A_i$. As $\delta_{A_i}$ is a principal eigenvector of $A_i$, we have $\delta_{A_i}^T A_i \delta_{A_i} = \delta_{A_i}^T \delta_{A_i} \|A_i\| = \|A_i\|$.
    \end{itemize}
    We now evaluate the product $\delta_h^T \Gamma_h \delta_h$.
    \begin{align}
        \delta_h^T \Gamma_h \delta_h &= \sum_{a, b \in S_f} \left(\delta_f[a] \cdot \bigotimes_{i=1}^k \delta_{A_i}\right)^T \left(\Gamma_f[a, b] \cdot \bigotimes_{i=1}^k \Gamma_{g_i}^{(a_i, b_i)}\right) \left(\delta_f[a] \cdot \bigotimes_{i=1}^k \delta_{A_i}\right) \\
        &= \sum_{a, b \in S_f} \left(\delta_f[a] \Gamma_f[a, b] \delta_f[b]\right) \cdot \left(\bigotimes_{i=1}^k \delta_{A_i}\right)^T \left(\bigotimes_{i=1}^k \Gamma_{g_i}^{(a_i, b_i)}\right) \left(\bigotimes_{i=1}^k \delta_{A_i}\right) \\
        &= \sum_{a, b \in S_f} \delta_f[a] \Gamma_f[a, b] \delta_f[b] \prod_{i=1}^k \|A_i\| \\
        &= \|\Gamma_f\| \cdot \prod_{i=1}^k \|A_i\|
    \end{align}
    Where the last equality follows because $\delta_f$ is a principal eigenvector of $\Gamma_f$. Therefore, we have $\|\Gamma_h\| \geq \|\Gamma_f\| \cdot \prod_{i=1}^k \|A_i\|$, completing the proof.
\end{proof}

\compositeMatrixDenominatorBound*

\begin{proof}
    We will apply \cref{lem:composite-function-matrix-construction}. To do so, we need to show that $\Gamma_h \circ D^h_i$ is the composition adversary matrix generated by $\Gamma_f \circ D^f_p$, the $A_d$ for $d \neq p$, and $A_p \circ D^{A_p}_q$. To that end, for $a, b \in \Sigma$ let $\Gamma^{(a, b)}_{(g_p, q)}$ be the matrix:
    \begin{align}
        \Gamma_{(g_p, q)}^{(a, b)} &= \begin{cases}
            \|A_p \circ D^{A_p}_q\|I & \text{if $a=b$} \\
            A_p \circ D^{A_p}_q & \text{otherwise}
        \end{cases}
    \end{align}
    Then we claim that, for all $x, y \in S_h$:
    \begin{align}
        \label{eq:composite-matrix-respects-D-conclusion}
        &(\Gamma_h \circ D^h_i)[x, y] \\
        =& (\Gamma_f \circ D^f_p)[\tilde{x}, \tilde{y}] \cdot \left(\left(\bigotimes_{d < p} \Gamma_{g_d}^{(\tilde{x}_d, \tilde{y}_d)}\right) \otimes \Gamma_{(g_p, q)}^{(\tilde{x}_p, \tilde{y}_p)} \otimes \left(\bigotimes_{d > p} \Gamma_{g_d}^{(\tilde{x}_d, \tilde{y}_d)}\right)\right)[x, y] \notag
    \end{align}
    We will show \eqref{eq:composite-matrix-respects-D-conclusion} by splitting into several cases based on $x$ and $y$ and their relation to $i$.
    \begin{itemize}
        \item \textbf{Case $x_i=y_i$ and $\tilde{x}_p = \tilde{y}_p$.} Then $D^h_i[x, y] = 0$ and $D_p^f[\tilde{x}, \tilde{y}] = 0$, so both sides of \eqref{eq:composite-matrix-respects-D-conclusion} are zero.
        \item \textbf{Case $x_i=y_i$ and $\tilde{x}_p \neq \tilde{y}_p$.} As in the previous case, we have $D^h_i[x, y] = 0$. Since $x_i = x^p_q$ and $y_i = y^p_q$, we also have $D^{A_p}_q[x^p, y^p] = 0$. Since $\tilde{x}_p \neq \tilde{y}_p$ we have $\Gamma^{(\tilde{x}_p, \tilde{y}_p)}_{(g_p, q)} = A_p \circ D^{A_p}_q$, so both sides of \eqref{eq:composite-matrix-respects-D-conclusion} are zero.
        \item \textbf{Case $x_i \neq y_i$ and $\tilde{x}_p = \tilde{y}_p$.} Then $D^f_p[\tilde{x}, \tilde{y}] = 0$. One component of the tensor product making up $\Gamma_h$ is $\Gamma_{g_p}^{(\tilde{x}_p, \tilde{y}_p)}$, which equals $\|A_p\|I$ because $\tilde{x}_p = \tilde{y}_p$. Since $x_i \neq y_i$, we have $x^p \neq y^p$, so $\Gamma_h[x, y] = 0$ and again both sides of \eqref{eq:composite-matrix-respects-D-conclusion} are zero.
        \item \textbf{Case $x_i \neq y_i$ and $\tilde{x}_p \neq \tilde{y}_p$.} Then $D_i^h[x, y] = D^f_p[\tilde{x}, \tilde{y}] = 1$, so they can both be ignored. Since $\tilde{x}_p \neq \tilde{y}_p$, we have $\Gamma_{(g_p, q)}^{(\tilde{x}_p, \tilde{y}_p)} = A_p \circ D_q^{A_p}$. But we also have $x_i \neq y_i$, so $x^p \neq y^p$ and therefore:
        \begin{align}
            \Gamma^{(\tilde{x}_p, \tilde{y}_p)}_{(g_p, q)}[x^p, y^p] = (A_p \circ D^{A_p}_q)[x^p, y^p] = A_p[x^p, y^p] = \Gamma^{(\tilde{x}_p, \tilde{y}_p)}_{g_p}[x^p, y^p]\,.
        \end{align}
        Therefore, both sides of \eqref{eq:composite-matrix-respects-D-conclusion} reduce to $\Gamma_h[x, y]$.
    \end{itemize}
    In all cases \eqref{eq:composite-matrix-respects-D-conclusion} holds. Applying \cref{lem:composite-function-matrix-construction} then completes the proof.
\end{proof}

\section{Ordered Search Deferred Proofs} \label{app:ordered-search-deferred-proofs}

In this section we give the proofs omitted from \cref{sec:ordered-search}. We start with a technical lemma that is used for both ordered search variants.

\begin{lemma}
    \label{lem:hilbert-matrix-bounds}
    For $m \in \mathbb{N}$, let $A_m \in \mathbb{R}^{m \times m}$ denote the matrix with entries $A_m[i, j] = \frac{1}{|i-j|+1}$. For $i \in [m]$, let $D_i^{A_m}$ denote the matrix with entries:
    \begin{align}
        D_i^{A_m}[x, y] &= \begin{cases}
            1 & \text{if $x \leq i \leq y$ or $y \leq i \leq x$} \\
            0 & \text{otherwise}
        \end{cases}
    \end{align}
    Then we have:
    \begin{align}
        \|A_m\| \in \Omega(\log m) \qquad \text{and} \qquad \|A_m \circ D_i^{A_m}\| \leq 2\pi
    \end{align}
\end{lemma}

The notation $D^{A_m}_i$ deliberately matches that of \cref{def:tile-distinguisher-matrix}, since they will later be distinguisher matrices for $HSOS_m$.

\begin{proof}
    We first show the $\Omega(\log m)$ lower bound on $\|A_m\|$ by computing $\mathbbold{1}^T A_m \mathbbold{1}$, where $\mathbbold{1}$ denotes the $m$-length all-ones vector:
    \begin{align}
        \mathbbold{1}^T A_m \mathbbold{1} &= \sum_{i=1}^m \sum_{j=1}^m \frac{1}{|i-j|+1} \\
        &\geq \sum_{i=1}^m \sum_{k=1}^{\max \{i, m-i+1\}} \frac{1}{k} \\
        &\geq \sum_{i=1}^m \sum_{k=1}^{\lceil m/2 \rceil} \frac{1}{k} \\
        &= \|\mathbbold{1}\|^2 \cdot \Theta(\log m)
    \end{align}
    Therefore, we have $\|A_m\| \in \Omega(\log m)$.

    For the element-wise product $A_m \circ D^{A_m}_i$, we first note that it has the form:
    \begin{align}
        A_m \circ D^{A_m}_i &= \begin{pmatrix}
            0 & \cdots & 0 & \frac{1}{i} & \frac{1}{i+1} & \cdots & \frac{1}{m} \\
            \vdots & \ddots & \vdots & \vdots & \vdots & \ddots & \vdots \\
            0 & \cdots & 0 & \frac{1}{2} & \frac{1}{3} & \cdots & \frac{1}{m-i+2} \\
            \frac{1}{i} & \cdots & \frac{1}{2} & 1 & \frac{1}{2} & \cdots & \frac{1}{m-i+1} \\
            \frac{1}{i+1} & \cdots & \frac{1}{3} & \frac{1}{2} & 0 & \cdots & 0 \\
            \vdots & \ddots & \vdots & \vdots & \vdots & \ddots & \vdots \\
            \frac{1}{m} & \cdots & \frac{1}{m-i+2} & \frac{1}{m-i+1} & 0 & \cdots & 0
        \end{pmatrix}
    \end{align}
    Where the central $1$ is in the $i$th row and column. 
    The upper-right and lower-left quadrants (with one row and column of overlap) are rotated submatrices of a Hilbert matrix $H_m$, defined by:
    \begin{align}
        H_m &= \begin{pmatrix}
            1 & \frac{1}{2} & \cdots & \frac{1}{m} \\
            \frac{1}{2} & \frac{1}{3} & \cdots & \frac{1}{m+1} \\
            \vdots & \vdots & \ddots & \vdots \\
            \frac{1}{m} & \frac{1}{m+1} & \cdots & \frac{1}{2m-1}
        \end{pmatrix}
    \end{align}
    Because $A_m \circ D_i^{A_m}$ is nonnegative and symmetric, it has a nonnegative principal eigenvector $\delta$. Let $\delta_{\ell} \in \mathbb{R}^{i}$ consist of the first $i$ entries of $\delta$, but in reverse order. Let $\delta_r \in \mathbb{R}^{m-i+1}$ consist of the last $m-i+1$ entries of $\delta$. Let $H_{i, m-i+1} \in \mathbb{R}^{i \times (m-i+1)}$ denote the top-left submatrix of $H_m$ of the appropriate dimensions. Then, upper-bounding $\delta^T(A_m \circ D_i^{A_m})\delta$ by summing the contributions from its upper-right and lower-left quadrants:
    \begin{align}
        \delta^T (A_m \circ D_i^{A_m}) \delta &\leq \delta_{\ell}^T H_{i, m-i+1} \delta_r + \delta_r^T H_{i, m-i+1}^T \delta_{\ell} = 2 \delta_{\ell}^T H_{i, m-i+1} \delta_r
    \end{align}
    Extend $\delta_{\ell}$ into $u \in \mathbb{R}^m$ by appending the last $m-i$ entries of $\delta$. Similarly, extend $\delta_r$ into $v \in \mathbb{R}^m$ by appending the first $i-1$ entries of $\delta$. Because $u$ and $v$ are permutations of $\delta$, we have $\|u\|=\|v\|=1$. Then:
    \begin{align}
        2 \delta_{\ell}^T H_{i, m-i+1} \delta_r &\leq 2 u^T H_m v \leq 2 \|u\| \cdot \|H_m\| \cdot \|v\| = 2\|H_m\|
    \end{align}
   So since $\|H_m\| \leq \pi$ (see, for example, \cite{choi1980hilbert}), we have $\|A_m \circ D_i^{A_m}\| \leq 2\pi$.
\end{proof}

We now prove each of our two lower bounds.

\hsosLowerBound*

\begin{proof}
    The function $HSOS_m$ is a generalized search function under the labeling where $(\sigma, j)$ refers to $\rightarrow^{j-1}\sigma \leftarrow^{m-j}$. We can then use the tile $A_m$ with entries:
    \begin{align}
        A_m[i, j] &= \frac{1}{|i-j|+1}
    \end{align}
    For each $i \in [m]$ and any two instances $(\sigma_1, x)$ and $(\sigma_2, y)$ for $\sigma_1 \neq \sigma_2$, the values of $(\sigma_1, x)_i$ and $(\sigma_2, y)_i$ differ if and only if character $i$ lies weakly between characters $x$ and $y$, so the distinguisher matrices for $A_m$ are:
    \begin{align}
        D^{A_m}_i[x, y] &= \begin{cases}
            1 & \text{if $x \leq i \leq y$ or $y \leq i \leq x$} \\
            0 & \text{otherwise}
        \end{cases}
    \end{align}
    These are exactly the same matrices as in \cref{lem:hilbert-matrix-bounds}, so we have $\|A_m\| \in \Omega(\log m)$ and $\|A_m \circ D^{A_m}_i\| \leq 2\pi$ for all $i \in [m]$. Therefore, \cref{lem:uniform-adversary-bound-in-terms-of-tile} gives us:
    \begin{align}
        SA(HSOS_m) &\geq \min_{i \in [m]} \frac{\|A_m\|}{\|A_m \circ D_i^{A_m}\|} \in \Omega(\log m)
    \end{align}
\end{proof}

\osLowerBound*

\begin{proof}
    We label each instance of $OS_m$ by its answer, so that $i \in [m]$ refers to $\uparrow^{i-1} * \downarrow^{m-i}$. We use the adversary matrix $\Gamma_{OS_m} \in \mathbb{R}^{m \times m}$ with entries:
    \begin{align}
        \Gamma_{OS_m}[x, y] &= \begin{cases}
            0 & \text{if $x=y$} \\
            \frac{1}{|x-y|+1} & \text{otherwise}
        \end{cases}
    \end{align}
    Let $A_m$ denote the matrix from \cref{lem:hilbert-matrix-bounds} with entries:
    \begin{align}
        A_m[x, y] &= \frac{1}{|x-y|+1}
    \end{align}
    We then have $\Gamma_{OS_m} = A_m - I$. Therefore, since $\|A_m\| \in \Omega(\log m)$:
    \begin{align}
        \|\Gamma_{OS_m}\| &= \|A_m - I\| \geq \|A_m\| - 1 \in \Omega(\log m) \label{eq:gamma-os-m-norm-greater-than-a-minus-one}
    \end{align}
    For each $i \in [m]$, the distinguisher matrix $D^{OS_m}_i$ has entries of the form:
    \begin{align}
        D^{OS_m}_i[x, y] &= \begin{cases}
            1 & \text{if $x \leq i \leq y$ or $y \leq i \leq x$} \\
            0 & \text{otherwise}
        \end{cases}
    \end{align}
    These are precisely the matrices $D_i^{A_m}$ from \cref{lem:hilbert-matrix-bounds}, so we have:
    \begin{align}
        \|\Gamma_{OS_m} \circ D_i^{OS_m}\| &\leq \|A_m \circ D_i^{A_m} \| \leq 2\pi \label{eq:gamma-os-m-times-d-norm-less-than-2-pi}
    \end{align}
    Combining \eqref{eq:gamma-os-m-norm-greater-than-a-minus-one} and \eqref{eq:gamma-os-m-times-d-norm-less-than-2-pi}, we have:
    \begin{align}
        SA(OS_m) &\geq \min_{i \in [m]} \frac{\|\Gamma^{OS_m}\|}{\|\Gamma^{OS_m} \circ D_i^{A_m}\|} \geq \min_{i \in [m]} \frac{\|A_m\| - 1}{\|A_m \circ D_i^{A_m}\|} \in \Omega(\log m)\,,
    \end{align}
    Which completes the proof.
\end{proof}

\section{Reduction to $TARSKI(n, 2)$ Deferred Proofs} \label{app:reduction-deferred-proofs}

\straightLineValidSpine*

\begin{proof}
    We first show that it starts at $u$ and ends at $v$.  The interpolation that defines $(L(u, v, u_1 + u_2))_1$ reduces to $(L(u, v, u_1+u_2))_1=u_1$ and therefore $(L(u, v, u_1+u_2))_2=u_2$. On the other end, we get $(L(u, v, v_1+v_2))_1=v_1$ and therefore $(L(u, v, v_1+v_2))_2=v_2$.
    
    We now show that it is connected and monotone. We do this by showing for arbitrary $c \in \{1, 2, \ldots, 2n-1\}$ that $L(u, v, c+1) - L(u, v, c)$ is either $(1, 0)$ or $(0, 1)$. For $\alpha \in \mathbb{R}$, let $x(u, v, \alpha) = u_1 \frac{v_1+v_2-\alpha}{v_1+v_2-u_1-u_2} + v_1 \frac{\alpha-u_1-u_2}{v_1+v_2-u_1-u_2}$. Then since $u \leq v$, we have:
    \begin{align}
        x(u, v, c+1) - x(u, v, c) &= \frac{v_1-u_1}{(v_1-u_1) + (v_2-u_2)} \leq 1 \label{eq:x-u-v-c-diff-less-than-one}
    \end{align}
    Since $u_1 \leq v_1$ and $x(u, v, \alpha)$ uses $\alpha$ to linearly interpolate between $u_1$ and $v_1$, we also have that $x$ is nondecreasing in $\alpha$, so:
    \begin{align}
        x(u, v, c+1) - x(u, v, c) &\geq 0 \label{eq:x-u-v-c-diff-nonnegative}
    \end{align}
    The coordinates $(L(u, v, c))_1$ and $(L(u, v, c+1))_1$ are chosen by rounding $x(u, v, c)$ and $x(u, v, c+1)$, consistently breaking ties towards higher values. Combining the bounds of \eqref{eq:x-u-v-c-diff-less-than-one} and \eqref{eq:x-u-v-c-diff-nonnegative}, therefore, we have:
    \begin{align}
        (L(u, v, c+1))_1 - (L(u, v, c))_1 \in \{0, 1\}
    \end{align}
    But the coordinate sum of $L(u, v, c+1)$ is exactly one greater than the coordinate sum of $L(u, v, c)$, so the possible $x$-coordinate differences of $0$ and $1$ correspond to $y$-coordinate differences of $1$ and $0$, respectively. Therefore, we have:
    \begin{align}
        L(u, v, c+1) - L(u, v, c) \in \{(0, 1), (1, 0)\}
    \end{align}
    Which completes the proof.
\end{proof}

\straightLineMonotone*

\begin{proof}
    Let $x(u, v, c) = u_1 \frac{v_1+v_2-c}{v_1+v_2-u_1-u_2} + v_1 \frac{c-u_1-u_2}{v_1+v_2-u_1-u_2}$. Because we only consider $u$ such that $u_1+u_2=b$ and $v$ such that $v_1 + v_2=d$, this formula simplifies to:
    \begin{align}
        x(u, v, c) &= u_1 \frac{d - c}{d - b} + v_1 \frac{c - b}{d - b}
    \end{align}
    Now the only dependence of $x(u, v, c)$ on $u$ and $v$ comes from the single instances of $u_1$ and $v_1$. If $c \leq d$, then the coefficient of $u_1$ is nonnegative, so $x(u, v, c)$ is nondecreasing in $u_1$. Similarly, if $c \geq b$, then the coefficient of $v_1$ is nonnegative, so $x(u, v, c)$ is nondecreasing in $v_1$. Rounding $x(u, v, c)$ to get $(L(u, v, c))_1$ preserves these monotonicities, which completes the proof. 
\end{proof}

\pointSlidingContiguousRegions*

\begin{proof}
    First consider the case where $u \in S_1$ is fixed and $v \in S_2$ varies. By \cref{lem:straight-line-monotone} the value of $(L(u, v, c))_1$ is monotone increasing in $v_1$, so the range of values of $v_1$ for which $(L(u, v, c))_1 < x$, $(L(u, v, c))_1 = x$, and $(L(u, v, c))_1 > x$ are each contiguous (or possibly empty). The dividing lines between these regions can be taken as $\delta_1$ and $\delta_4$. By \cref{lem:straight-line-valid-spine}, for $v_1 \in (\delta_1, \delta_4]$ we have:
    \begin{itemize}
        \item Either $L(u, v, c+1) = L(u, v, c)+(0, 1)$ or $L(u, v, c+1) = L(u, v, c)+(1, 0)$.
        \item Either $L(u, v, c-1) = L(u, v, c)+(-1, 0)$ or $L(u, v, c-1) = L(u, v, c)+(0, -1)$.
    \end{itemize}
    So by applying \cref{lem:straight-line-monotone} to $(L(u, v, c+1))_1$ and $(L(u, v, c-1))_1$, we get thresholds $\delta_2$ and $\delta_3$ that divide these contiguous regions too.

    Now consider the case where $v \in S_2$ is fixed and $u \in S_1$ varies. The same reasoning applies: as $u_1$ varies, the values of $(L(u, v, c))_1$, $(L(u, v, c+1))_1$, and $(L(u, v, c-1))_1$ vary in parallel by \cref{lem:straight-line-monotone}, so we get thresholds $\delta_1, \delta_2, \delta_3, \delta_4$ between contiguous regions.
\end{proof}

\regionBoundarySizeN*

\begin{proof}
    We will show \eqref{eq:B-elements-characterization} first, which then implies \eqref{eq:B-set-size-equals-n}. There are three constraints that define points $(x, y) \in B^{2(n-1)i+n+1}$:
    \begin{enumerate}[(i)]
        \item $x+y=2(n-1)i+n+1$,
        \item $|x-y| \leq n$, and
        \item $(x, y) \in [n']^2=[n(n^2+n-1)]^2$.
    \end{enumerate}
    The points that satisfy (i) are precisely those of the form:
    \begin{align}
        B^{2(n-1)i+n+1} \subseteq \left\{\bigl((n-1)i+j, (n-1)i+n+1-j\bigr) \mid j \in \mathbb{Z} \right\} \label{eq:B-set-subset-j-parameterization}
    \end{align}
    So all that remains is to show that the only values of $j$ that satisfy the other constraints are $j \in [n]$. In fact, (ii) by itself makes this restriction, since applying it to points of the form in \eqref{eq:B-set-subset-j-parameterization} gives:
    \begin{align}
        |2j - (n+1)| \leq n-1 \label{eq:2j-minus-n-minus-1-absolute-value}
    \end{align}
    Or, equivalently after breaking up \eqref{eq:2j-minus-n-minus-1-absolute-value} into two linear inequalities:
    \begin{align}
        j \leq n \quad \text{and} \quad 1 \leq j
    \end{align}
    All that remains is to show that these points satisfy (iii). Without loss of generality we show that their $x$-coordinates lie in $[n']$, since their $y$-coordinates cover the same range. Since $i \geq 0$ and $j \geq 1$, we have:
    \begin{align}
        (n-1)i+j &\geq 1
    \end{align}
    And since $i \leq n(n+2)$ and $j \leq n$, we have:
    \begin{align}
        (n-1)i+j &\leq (n-1)n(n+2)+n = n(n^2+n-1)
    \end{align}
    Which completes the proof.
\end{proof}

\regionShapeRectangular*

\begin{proof}
    We claim that the desired $\ell$ is:
    \begin{align}
        \ell &= w_1 - \left\lfloor \frac{w_1 + w_2 - (n+1)}{2} \right\rceil
    \end{align}
    We first show that $\ell \in [n]$. To do so, let $a, b \in [n]$ be such that $w$ is on the grid line from $B^{\textsc{Low}(i, j)}_{a}$ to $B^{\textsc{High}(i, j)}_b$. We then have $w_1 = (L(B^{\textsc{Low}(i, j)}_{a}, B^{\textsc{High}(i, j)}_b, w_1+w_2))_1$. Expanding out the definition of $L$ and the expressions for $B^{\textsc{Low}(i, j)}_{a}$ and $B^{\textsc{High}(i, j)}_b$ from \cref{lem:region-boundary-size-n}:
    \begin{align}
        w_1 =& \left\lfloor \left(\frac{\textsc{Low}(i, j)-(n+1)}{2}+a\right)\frac{\textsc{High}(i, j)-(w_1+w_2)}{\textsc{High}(i, j)-\textsc{Low}(i, j)}\right. \notag \\
        &+\left. \left(\frac{\textsc{High}(i, j)-(n+1)}{2}+b\right)\frac{(w_1+w_2)-\textsc{Low}(i, j)}{\textsc{High}(i, j) - \textsc{Low}(i, j)}\right\rceil \\
        =& \left\lfloor \frac{w_1+w_2-(n+1)}{2} \right. \notag \\
        &+ \left.a\frac{\textsc{High}(i, j)-(w_1+w_2)}{\textsc{High}(i, j)-\textsc{Low}(i, j)} + b\frac{(w_1+w_2) - \textsc{Low}(i, j)}{\textsc{High}(i, j)-\textsc{Low}(i, j)}\right\rceil \label{eq:w1-expression-with-a-b-average}
    \end{align}
    By \cref{lem:straight-line-valid-spine}, vertex $w$ satisfies $B^{\textsc{Low}(i, j)}_a \leq w \leq B^{\textsc{High}(i, j)}_b$, so the second line of \eqref{eq:w1-expression-with-a-b-average} is a weighted average of $a$ and $b$. Therefore, it lies in the interval $[1, n]$. Applying each of these inequalities in turn, we get:
    \begin{align}
        w_1 \geq \left\lfloor \frac{w_1+w_2-(n+1)}{2} \right\rceil + 1 \quad \text{and} \quad w_1 \leq \left\lfloor \frac{w_1+w_2-(n+1)}{2} \right\rceil + n \notag
    \end{align}
    Which combine to give $\ell \in [n]$.

    We now show that this choice of $\ell$ puts $w$ on the correct grid lines. Let $p, q \in \mathbb{N}$ be such that $\textsc{Low}(c_1, d_1) = 2(n-1)p+n+1$ and $\textsc{High}(c_2, d_2) = 2(n-1)q+n+1$. We claim that:
    \begin{align}
        L(B^{2(n-1)p+n+1}_{\ell}, B^{2(n-1)q+n+1}_{\ell}, w_1+w_2) = w
    \end{align}
    So $w$ would be on this grid line. To do so, we explicitly compute the $x$-coordinate from its definition:
    \begin{align}
        & (L(B^{2(n-1)p+n+1}_{\ell}, B^{2(n-1)q+n+1}_{\ell}, w_1+w_2))_1 \\
        =& \left\lfloor ((n-1)p+\ell) \frac{2(n-1)q+n+1-(w_1+w_2)}{2(q-p)(n-1)} \right. \notag \\
        &+ \left.((n-1)q+\ell)\frac{(w_1+w_2)-2(n-1)p-n-1}{2(q-p)(n-1)}\right\rceil \\
        =& \left\lfloor \ell + \frac{w_1+w_2 - (n+1)}{2} \right\rceil = w_1
    \end{align}
    And since the $y$-coordinate is chosen so that the coordinates sum to $w_1+w_2$, its $y$-coordinate is $w_2$ as desired.
\end{proof} 

\tarskiNOSReduction*

\begin{proof}
    Let $n' = n(n^2+n-1)$. We will consider only instances of $TARSKI(n', 2)$ in the set $\mathcal{T}(n')$ from \cref{def:tarski-instances}. Each such function is parameterized by a choice of $C \in [n]^{n+1}$ and $i \in [n+1]$. These parameters are also enough to specify an instance of $NOS_{n+1, n}$: the answer locations of the $n+1$ inner $HSOS_n$ instances can be given by $C$, and the answer to the outer $OS_{n+1}$ instance can be given by $i$. Using this one-to-one correspondence, let $\Gamma_{NOS}$ be an optimal adversary matrix for $NOS_{n+1, n}$ and let $\Gamma_{\text{Tarski}}=\Gamma_{NOS}$ be an adversary matrix for the $TARSKI(n', 2)$ instances in $\mathcal{T}(n')$. This is a valid adversary matrix because whenever two $TARSKI$ instances have the same answer, they must have the same $i$, so the corresponding $NOS_{n+1, n}$ instances also have the same answer.

    Because we use exactly the same matrix, we have $\|\Gamma_{\text{Tarski}}\| = \|\Gamma_{NOS}\|$. All that remains is to analyze the denominator of the spectral adversary method.

    For each query location $(x, y)$, let $D^{\text{Tarski}}_{(x, y)}$ be the distinsuigher matrix:
    \begin{equation}
        D^{\text{Tarski}}_{(x, y)}[f, g] = \begin{cases}
            1 & \text{if $f(x, y) \neq g(x, y)$} \\
            0 & \text{otherwise}
        \end{cases}
    \end{equation}
    For instances $f$ and $g$ of $NOS_{n+1, n}$, arbitrary $i \in [n+1]$, and arbitrary $j \in [n]$, let $D^{NOS}_{i, j}$ be the distinguisher matrix:
    \begin{equation}
        D^{NOS}_{i, j}[f, g] = \begin{cases}
            1 & \text{if the $j$th character of the $i$th block of $f$ and $g$ differ} \\
            0 & \text{otherwise}
        \end{cases}
    \end{equation}

    The matrices $D^{\text{Tarski}}_{(x, y)}$ and $D^{NOS}_{i, j}$ are very closely connected when $(x, y)$ is on a chunk boundary. Specifically, we claim:
    \begin{align}
        D^{\text{Tarski}}_{B^{\textsc{Bound}(i)}_j} = D^{NOS}_{i, j} \quad \forall i \in [n+1], j \in [n] \label{eq:tarski-D-equals-NOS-for-chunk-boundary}
    \end{align}

    We prove this claim by making a one-to-one correspondence between the responses to querying $(i, j)$ in $NOS_{n+1, n}$ and querying $B^{\textsc{Bound}(i)}_j$ in $TARSKI(n', 2)$. Let $C \in [n]^{n+1}$ and $\ell \in [n+1]$ be the parameters of an element of $\mathcal{T}(n')$ and its corresponding instance of $NOS_{n+1, n}$. We consider several cases based on the roles $(i, j)$ plays in these instances.
    \begin{itemize}
        \item If $j=C_i$ and $i=\ell$, then in $NOS_{n+1, n}$ this query location is the uniquely identifiable $*$. By the construction of the chunked spine, the fixed point of the $TARSKI$ instance is $B^{\textsc{Bound}(i)}_j$, which is also uniquely identifiable. Both values are unique in their respective problems, so they are distinguished from all other query locations.
        \item If $j=C_i$ but $i \neq \ell$, then in the $NOS_{n+1, n}$ instance the $(i, j)$ query location is either $\uparrow$ or $\downarrow$ depending on whether $\ell>i$ or $\ell<i$. By the construction of the chunked spine, the vertex $B^{\textsc{Bound}(i)}_j$ is on the spine. The spine leading up to it is a grid line at a predictable angle (either from $(1, 1)$ if $i=1$ or from $B^{\textsc{High}(i-1, n+1)}_j$ if $i > 1$) and so is the spine following it (either towards $B^{\textsc{Low}(i, 2)}_j$ if $i<n+1$ or towards $(n', n')$ if $i=n+1$). Therefore, the value at $B^{\textsc{Bound}(i)}_j$ is determined entirely by whether $\ell>i$ or $\ell<i$, exactly like the $(i, j)$ query location in $NOS_{n+1, n}$. Both possible values are also different from the $i=\ell$ case.
        \item If $j \neq C_i$, then in the $NOS_{n+1, n}$ instance the $(i, j)$ query location is either $\rightarrow$ or $\leftarrow$ depending on whether $j<C_i$ or $j>C_i$. The chunked spine passes through $B^{\textsc{Bound}(i)}_{C_i}$ rather than $B^{\textsc{Bound}(i)}_j$, so the value at $B^{\textsc{Bound}(i)}_j$ is entirely determined on whether it is up-left or down-right of $B^{\textsc{Bound}(i)}_{C_i}$. These cases correspond exactly to $j<C_i$ and $j>C_i$, and are distinct from the values in all $j=C_i$ cases.
    \end{itemize}
    Therefore, the $TARSKI$ instances distinguished by querying $B^{\textsc{Bound}(i)}_j$ correspond exactly to the $NOS_{n+1, n}$ instances distinguished by querying $(i, j)$, so $D^{\text{Tarski}}_{B^{\textsc{Bound}(i)}_j} = D^{NOS}_{i, j}$, proving \eqref{eq:tarski-D-equals-NOS-for-chunk-boundary}.

    We now return to the spectral adversary method by claiming:
    \begin{equation}
        \|\Gamma_{\text{Tarski}} \circ D^{\text{Tarski}}_{(x, y)}\| \leq 7 \cdot \max_{i, j} \|\Gamma_{NOS} \circ D^{NOS}_{i, j}\| \quad \forall (x, y) \in [n']^2 \label{eq:max-tarski-product-less-than-fifteen-max-nos-product}
    \end{equation}
    If \eqref{eq:max-tarski-product-less-than-fifteen-max-nos-product} holds, we have:
    \begin{align}
        SA(TARSKI(n', 2)) &\geq \frac{\|\Gamma_{\text{Tarski}}\|}{\max_{(x, y) \in [n']^2} \|\Gamma_{\text{Tarski}} \circ D^{\text{Tarski}}_{(x, y)}\|} \\
        &\geq \frac{1}{7} \cdot \frac{\|\Gamma_{NOS}\|}{\max_{i \in [n+1], j \in [n]} \|\Gamma_{NOS} \circ D^{NOS}_{i, j}\|} \\
        &= \frac{1}{7}SA(NOS_{n+1, n})
    \end{align}
    Which would complete the proof. To prove \eqref{eq:max-tarski-product-less-than-fifteen-max-nos-product}, we consider several cases for $(x, y)$.

    \paragraph{Case $(x, y)$ up-left or down-right of the chunks.} Since a chunked spine always runs through the chunks, the up-left triangle is always above the spine and the down-right triangle is always below it. In the herringbone construction, this alone determines the function values, so no instances differ here. Therefore, every entry of $\Gamma_{\text{Tarski}} \circ D^{\text{Tarski}}_{(x, y)}$ is zero, so $\|\Gamma_{\text{Tarski}} \circ D^{\text{Tarski}}_{(x, y)}\| = 0$, which proves \eqref{eq:max-tarski-product-less-than-fifteen-max-nos-product} in this case.

    \paragraph{Case $(x, y)$ on a chunk boundary.} Then there exists $i \in [n+1]$ and $j \in [n]$ such that $(x, y) = B^{\textsc{Bound}(i)}_j$. We can appeal directly to the choice of $\Gamma_{\text{Tarski}} = \Gamma_{NOS}$ and to \eqref{eq:tarski-D-equals-NOS-for-chunk-boundary} to get:
    \begin{align}
        \|\Gamma_{\text{Tarski}} \circ D^{\text{Tarski}}_{(x, y)}\| = \|\Gamma_{NOS} \circ D^{NOS}_{i, j}\|\,,
    \end{align}
    Which proves \eqref{eq:max-tarski-product-less-than-fifteen-max-nos-product} in this case.

    \paragraph{Case $x+y \leq n$.} Here $(x, y)$ is in the portion of $[n']^2$ that the spine runs through from $(1, 1)$ to the first chunk boundary. Let $f, g$ be arbitrary instances of $TARSKI(n', 2)$ following our construction such that $f(x, y) \neq g(x, y)$. This and the remaining cases will be handled by using \cref{lem:point-sliding-contiguous-regions} to find a small set $V$, independent of $f$ and $g$, for which we must have $f(v) \neq g(v)$ for some $v \in V$. Here we use \cref{lem:point-sliding-contiguous-regions} immediately with $S_1=\{(1, 1)\}$, $S_2=B^{\textsc{Bound}(1)}$, and the point $(x, y)$ to get thresholds $\delta_1, \delta_2, \delta_3, \delta_4$. We then choose:
    \begin{align}
        V &= \left\{v \in B^{\textsc{Bound}(1)} \mid v_1 \in \{\delta_1, \delta_2, \delta_4\}\right\}
    \end{align}
    Then $|V| \leq 3$. Suppose for contradiction that $f(v)=g(v)$ for all $v \in V$. Since the function values in $B^{\textsc{Bound}(1)}$ give ordered search feedback for $C^f_1$ and $C^g_1$, the spines of $f$ and $g$ fall into the same cases of \cref{lem:point-sliding-contiguous-regions}. Specifically:
    \begin{itemize}
        \item If $\left(B^{\textsc{Bound}(1)}_{C^f_1}\right)_1 \leq \delta_1$, then because $f(v)=g(v)$ for all $v \in V$ the same must hold for $\left(B^{\textsc{Bound}(1)}_{C^g_1}\right)_1$. Therefore by \cref{lem:point-sliding-contiguous-regions}, the spines of $f$ and $g$ both pass above $(x, y)$, so we would have $f(x, y)=g(x, y)=(x-1, y+1)$.
        \item If $\left(B^{\textsc{Bound}(1)}_{C^f_1}\right)_1 \in (\delta_1, \delta_2]$, then the same must hold for $\left(B^{\textsc{Bound}(1)}_{C^g_1}\right)_1$. By \cref{lem:point-sliding-contiguous-regions}, the spines of $f$ and $g$ both pass through $(x, y)$ and $(x, y+1)$. Because the fixed points of $f$ and $g$ have coordinate sums at least $\textsc{Bound}(1)$, their spines are oriented up-right near $(x, y)$, so we would have $f(x, y)=g(x, y)=(x, y+1)$.
        \item If $\left(B^{\textsc{Bound}(1)}_{C^f_1}\right)_1 \in (\delta_2, \delta_4]$, then the same must hold for $\left(B^{\textsc{Bound}(1)}_{C^g_1}\right)_1$. By \cref{lem:point-sliding-contiguous-regions}, the spines of $f$ and $g$ both pass through $(x, y)$ and $(x+1, y)$. Since both functions' spines are oriented up-right near $(x, y)$, we would have $f(x, y)=g(x, y)=(x+1,y)$.
        \item If $\left(B^{\textsc{Bound}(1)}_{C^f_1}\right)_1 > \delta_4$, then the same must hold for $\left(B^{\textsc{Bound}(1)}_{C^g_1}\right)_1$. By \cref{lem:point-sliding-contiguous-regions}, the spines of $f$ and $g$ both pass below $(x, y)$, so we would have $f(x, y)=g(x, y)=(x+1,y-1)$.
    \end{itemize}
    All cases contradict our assumption that $f(x, y) \neq g(x, y)$, so we must have $f(v) \neq g(v)$ for some $v \in V$. Because this holds for arbitrary instances $f$ and $g$, we have:
    \begin{align}
        D^{\text{Tarski}}_{(x, y)} &\leq \sum_{v \in V} D^{\text{Tarski}}_v, \label{eq:bottom-left-triangle-case-d-tarski-matrix-inequality}
    \end{align}
    Where the inequality is taken elementwise. Because $\Gamma_{\text{Tarski}}$ is nonnegative, taking the elementwise product of both sides of \eqref{eq:bottom-left-triangle-case-d-tarski-matrix-inequality} with $\Gamma_{\text{Tarski}}$ preserves the inequality:
    \begin{align}
        \Gamma_{\text{Tarski}} \circ D^{\text{Tarski}}_{(x, y)} &\leq \sum_{v \in V} \Gamma_{\text{Tarski}} \circ D^{\text{Tarski}}_v
    \end{align}
    Taking the norm of both sides and applying the triangle inequality gives:
    \begin{align}
        \left\| \Gamma_{\text{Tarski}} \circ D^{\text{Tarski}}_{(x, y)}\right\| &\leq \sum_{v \in V} \left\|\Gamma_{\text{Tarski}} \circ D^{\text{Tarski}}_v\right\| \\
        &\leq 3 \cdot \max_{v \in V} \left\|\Gamma_{\text{Tarski}} \circ D^{\text{Tarski}}_v\right\|
    \end{align}
    Using \eqref{eq:tarski-D-equals-NOS-for-chunk-boundary} and the choice of $\Gamma_{\text{Tarski}}=\Gamma_{NOS}$, we get:
    \begin{align}
        \left\| \Gamma_{\text{Tarski}} \circ D^{\text{Tarski}}_{(x, y)}\right\| &\leq 3 \cdot \max_{i \in [n+1], j \in [n]} \left\|\Gamma_{NOS} \circ D^{NOS}_{i, j}\right\|
    \end{align}
    Which proves \eqref{eq:max-tarski-product-less-than-fifteen-max-nos-product} in this case.

    \paragraph{Case $x+y \geq n'-n$.} Here $(x, y)$ is in the portion of $[n']^2$ that the spine runs through from the last chunk boundary to $(n', n')$. The proof here is extremely similar to the $x+y \leq n$ case. Let $f, g$ be arbitrary instances of $TARSKI(n', 2)$ following our construction such that $f(x, y) \neq g(x, y)$. Let $\delta_1, \delta_2, \delta_3, \delta_4$ be the thresholds given by \cref{lem:point-sliding-contiguous-regions} with $S_1=B^{\textsc{Bound}(n+1)}$, $S_2=\{(n', n')\}$, and the point $(x, y)$. We then choose:
    \begin{align}
        V &= \left\{v \in B^{\textsc{Bound}(n+1)} \mid v_1 \in \{\delta_1, \delta_3, \delta_4\}\right\}
    \end{align}
    Then $|V| \leq 3$. Suppose for contradiction that $f(v)=g(v)$ for all $v \in V$. Since the function values in $B^{\textsc{Bound}(n+1)}$ give ordered search feedback for $C^f_{n+1}$ and $C^g_{n+1}$, the spines of $f$ and $g$ fall into the same cases of \cref{lem:point-sliding-contiguous-regions}. Specifically:
    \begin{itemize}
        \item If $\left(B^{\textsc{Bound}(n+1)}_{C^f_{n+1}}\right)_1 \leq \delta_1$, then because $f(v)=g(v)$ for all $v \in V$ the same must hold for $\left(B^{\textsc{Bound}(n+1)}_{C^g_{n+1}}\right)_1$. Therefore by \cref{lem:point-sliding-contiguous-regions}, the spines of $f$ and $g$ both pass above $(x, y)$, so we would have $f(x, y)=g(x, y)=(x-1, y+1)$.
        \item If $\left(B^{\textsc{Bound}(n+1)}_{C^f_{n+1}}\right)_1 \in (\delta_1, \delta_3]$, then the same must hold for $\left(B^{\textsc{Bound}(n+1)}_{C^g_{n+1}}\right)_1$. By \cref{lem:point-sliding-contiguous-regions}, the spines of $f$ and $g$ both pass through $(x, y)$ and $(x-1, y)$. Because the fixed points of $f$ and $g$ have coordinate sums at most $\textsc{Bound}(n+1)$, their spines are oriented down-left near $(x, y)$, so we would have $f(x, y)=g(x, y)=(x-1, y)$.
        \item If $\left(B^{\textsc{Bound}(n+1)}_{C^f_{n+1}}\right)_1 \in (\delta_3, \delta_4]$, then the same must hold for $\left(B^{\textsc{Bound}(n+1)}_{C^g_{n+1}}\right)_1$. By \cref{lem:point-sliding-contiguous-regions}, the spines of $f$ and $g$ both pass through $(x, y)$ and $(x, y-1)$. Since both functions' spines are oriented down-left near $(x, y)$, we would have $f(x, y)=g(x, y)=(x,y-1)$.
        \item If $\left(B^{\textsc{Bound}(n+1)}_{C^f_{n+1}}\right)_1 > \delta_4$, then the same must hold for $\left(B^{\textsc{Bound}(n+1)}_{C^g_{n+1}}\right)_1$. By \cref{lem:point-sliding-contiguous-regions}, the spines of $f$ and $g$ both pass below $(x, y)$, so we would have $f(x, y)=g(x, y)=(x+1,y-1)$.
    \end{itemize}
    All cases contradict the assumption that $f(x, y) \neq g(x, y)$, so we must have $f(v) \neq g(v)$ for some $v \in V$. Because this holds for arbitrary instances $f$ and $g$, we have:
    \begin{align}
        D^{\text{Tarski}}_{(x, y)} &\leq \sum_{v \in V} D^{\text{Tarski}}_v,
    \end{align}
    Where the inequality is taken elementwise. By identical reasoning to the $x+y \leq n$ case, we eventually get:
    \begin{align}
        \left\| \Gamma_{\text{Tarski}} \circ D^{\text{Tarski}}_{(x, y)}\right\| &\leq |V| \cdot \max_{i \in [n+1], j \in [n]} \left\| \Gamma_{NOS} \circ D^{NOS}_{i, j}\right\| ,
    \end{align}
    Which, since $|V| \leq 3$, proves \eqref{eq:max-tarski-product-less-than-fifteen-max-nos-product} in this case.
    
    \paragraph{Case $(x, y)$ in a chunk.} All that remains are the points in a chunk but not on a chunk boundary. Let $\alpha \in [n], \beta \in [n+2]$ be such that $(x, y)$ is in region $\beta$ of chunk $\alpha$. Let $\gamma \in [n]$ be the index guaranteed by \cref{lem:region-shape-rectangular} such that $(x, y)$ is on the grid line from $B^{\textsc{Bound}(\alpha)}_{\gamma}$ to $B^{\textsc{Bound}(\alpha+1)}_{\gamma}$. We claim that there exists a set of indices $K \subseteq [n]$ of size at most $4$ such that, for all instances $f$ and $g$ of $TARSKI(n', 2)$ following our construction such that $f(x, y) \neq g(x, y)$, at least one of the following is true:
    \begin{enumerate}[(i)]
        \item $f(B^{\textsc{Bound}(\alpha)}_{\gamma}) \neq g(B^{\textsc{Bound}(\alpha)}_{\gamma})$,
        \item $f(B^{\textsc{Bound}(\alpha+1)}_{\gamma}) \neq g_(B^{\textsc{Bound}(\alpha+1)}_{\gamma})$,
        \item $f(B^{\textsc{Bound}(\alpha)}_{\beta-1}) \neq g(B^{\textsc{Bound}(\alpha)}_{\beta-1})$,
        \item or $f(B^{\textsc{Bound}(\alpha+1)}_i) \neq g(B^{\textsc{Bound}(\alpha+1)}_i)$ for some $i \in K$.
    \end{enumerate}
    These at most $7$ points will be our $V$ for this case.
    To show this, assume that for some $f, g$ all of (i)-(iii) are false. The desired $K$ will then be constructed, independently of $f$ and $g$. Let $C^f, C^g \in [n]^{n+1}$ be the generators of the chunked spines of $f$ and $g$, respectively, and assume without loss of generality that $C^f_\alpha \leq C^g_\alpha$.

    If $\beta<C^f_{\alpha}+1$, then by the construction of the chunked spines, in region $\beta$ the spines of $f$ and $g$ would be along grid lines from $B^{\textsc{Bound}(\alpha)}_{C^f_\alpha}$ to $B^{\textsc{High}(\alpha, \beta)}_{C^f_\alpha}$ and $B^{\textsc{Bound}(\alpha)}_{C^g_\alpha}$ to $B^{\textsc{High}(\alpha, \beta)}_{C^g_\alpha}$, respectively. But since $f(B^{\textsc{Bound}(\alpha)}_\gamma) = g(B^{\textsc{Bound}(\alpha)}_\gamma)$, we have that $C^f_\alpha$ and $C^g_\alpha$ are on the same side of $\gamma$, so their spines would be on the same side of $(x, y)$. This would contradict our assumption that $f(x, y) \neq g(x, y)$, so we have $\beta \geq C^f_{\alpha}+1$.
    
    Similarly, if $\beta > C^g_\alpha+1$, then in region $\beta$ the spines of $f$ and $g$ would be along grid lines from $B^{\textsc{Low}(\alpha, \beta)}_{C^f_{\alpha+1}}$ to $B^{\textsc{Bound}(\alpha+1)}_{C^f_{\alpha+1}}$ and $B^{\textsc{Low}(\alpha, \beta)}_{C^g_{\alpha+1}}$ to $B^{\textsc{Bound}(\alpha+1)}_{C^g_{\alpha+1}}$, respectively. Since $f(B^{\textsc{Bound}(\alpha+1)}_\gamma) = g(B^{\textsc{Bound}(\alpha+1)}_\gamma)$, we have that $C^f_{\alpha+1}$ and $C^g_{\alpha+1}$ are on the same side of $\gamma$, so their spines would be on the same side of $(x, y)$. This would again contradict our assumption that $f(x, y) \neq g(x, y)$, so we have $\beta \leq C^g_\alpha+1$.
    
    Therefore, we have $\beta-1 \in [C^f_\alpha, C^g_\alpha]$. However, we also have $f(B^{\textsc{Bound}(\alpha)}_{\beta-1}) = g(B^{\textsc{Bound}(\alpha)}_{\beta-1})$, so $C^f_\alpha$ and $C^g_\alpha$ must be on the same side of $\beta-1$. The only way for both of these conditions to hold is if $C^f_\alpha=C^g_\alpha=\beta-1$. Therefore, in region $\beta$, the spines of both $f$ and $g$ are transitioning along grid lines from $B^{\textsc{Low}(\alpha, \beta)}_{C^f_\alpha}$ to $B^{\textsc{High}(\alpha, \beta)}_{C^f_{\alpha+1}}$ and $B^{\textsc{Low}(\alpha, \beta)}_{C^g_\alpha}$ to $B^{\textsc{High}(\alpha, \beta)}_{C^g_{\alpha+1}}$, respectively. Moreover, since $B^{\textsc{Bound}(\alpha)}_{\beta-1}$ is on the spines of both $f$ and $g$ and $f(B^{\textsc{Bound}(\alpha)}_{\beta-1}) = g(B^{\textsc{Bound}(\alpha)}_{\beta-1})$, the orientation of the spine is the same near $(x, y)$; all that differs is the shape.
    
    We will choose $K$ to cover all relevant thresholds for that shape. Let $\delta_1, \delta_2, \delta_3, \delta_4$ be the thresholds given by \cref{lem:point-sliding-contiguous-regions} for $S_1=\{B^{\textsc{Low}(\alpha, \beta)}_{\beta-1}\}$, $S_2 = B^{\textsc{High}(\alpha, \beta)}$, and the point $(x, y)$. We then set:
    \begin{align}
        K = \bigcup_{i=1}^4 \left\{j \in [n] \mid (B^{\textsc{High}(\alpha, \beta)}_j)_1=\delta_i\right\}
    \end{align}
    Suppose for contradiction that $f(B^{\textsc{Bound}(\alpha+1)}_i)=g(B^{\textsc{Bound}(\alpha+1)}_i)$ for all $i \in K$. Then because queries in $B^{\textsc{Bound}(\alpha+1)}$ give ordered search feedback on $C^f_{\alpha+1}$ and $C^g_{\alpha+1}$, we have for all $i, j \in K$:
    \begin{align}
        C^f_{\alpha+1} \in (i, j] \iff C^g_{\alpha+1} \in (i, j]
    \end{align}
    By the choice of $K$ and because the spines of $f$ and $g$ pass through $B^{\textsc{High}(\alpha, \beta)}_{C^f_{\alpha+1}}$ and $B^{\textsc{High}(\alpha, \beta)}_{C^g_{\alpha+1}}$, respectively, we therefore have for all $i, j \in [4]$:
    \begin{align}
        \left(B^{\textsc{High}(\alpha, \beta)}_{C^f_{\alpha+1}}\right)_1 \in (\delta_i, \delta_j] \iff \left(B^{\textsc{High}(\alpha, \beta)}_{C^g_{\alpha+1}}\right)_1 \in (\delta_i, \delta_j]
    \end{align}
    Therefore, the spines of $f$ and $g$ fall into the same cases of \cref{lem:point-sliding-contiguous-regions}. The reasoning here is slightly more involved than in the $x+y \leq n$ and $x+y \geq n'-n$ cases, though, since the spines of $f$ and $g$ could both be going up-right or both be going down-left. For brevity, let $x'=(B^{\textsc{High}(\alpha, \beta)}_{C^f_{\alpha+1}})_1$. Going case-by-case:
    \begin{itemize}
        \item If $x' \leq \delta_1$, then by \cref{lem:point-sliding-contiguous-regions} the spines of $f$ and $g$ each pass up-left of $(x, y)$. Therefore $f(x, y)=g(x, y)=(x-1, y+1)$.
        \item If $x' \in (\delta_1, \delta_4]$, then the spines of $f$ and $g$ both pass through $(x, y)$. We now have several sub-cases based on the orientation of the spines of $f$ and $g$.
        \begin{itemize}
            \item If the spines of $f$ and $g$ go up-right through $(x, y)$ and $x' \in (\delta_1, \delta_2]$, then by \cref{lem:point-sliding-contiguous-regions} the spines of $f$ and $g$ both go from $(x, y)$ to $(x, y+1)$. Therefore $f(x, y)=g(x, y)=(x, y+1)$.
            \item If the spines of $f$ and $g$ go up-right through $(x, y)$ and $x' \in (\delta_2, \delta_4]$, then by \cref{lem:point-sliding-contiguous-regions} the spines of $f$ and $g$ both go from $(x, y)$ to $(x+1, y)$. Therefore $f(x, y)=g(x, y)=(x+1, y)$.
            \item If the spines of $f$ and $g$ go down-left through $(x, y)$ and $x' \in (\delta_1, \delta_3]$, then by \cref{lem:point-sliding-contiguous-regions} the spines of $f$ and $g$ both go from $(x, y)$ to $(x-1, y)$. Therefore $f(x, y)=g(x, y)=(x-1, y)$.
            \item If the spines of $f$ and $g$ go down-left through $(x, y)$ and $x' \in (\delta_3, \delta_4]$, then by \cref{lem:point-sliding-contiguous-regions} the spines of $f$ and $g$ both go from $(x, y)$ to $(x, y-1)$. Therefore $f(x, y)=g(x, y)=(x, y-1)$.
        \end{itemize}
        \item If $x' > \delta_4$, then by \cref{lem:point-sliding-contiguous-regions} the spines of $f$ and $g$ both pass down-right of $(x, y)$. Therefore $f(x, y)=g(x, y)=(x+1, y-1)$. 
    \end{itemize}
    In all cases $f(x, y)=g(x, y)$, contradicting the assumption that they differed. Therefore, we must have $f(B^{\textsc{Bound}(\alpha+1)}_i) \neq g(B^{\textsc{Bound}(\alpha+1)}_i)$ for some $i \in K$.

    We now have a set of points on chunk boundaries of size at most seven on which $f$ and $g$ cannot completely agree. These points are our $V$:
    \begin{align}
        V &= \left\{B^{\textsc{Bound}(\alpha)}_\gamma, B^{\textsc{Bound}(\alpha+1)}_\gamma, B^{\textsc{Bound}(\alpha)}_{\beta-1}\right\} \cup \left\{B^{\textsc{Bound}(\alpha+1)}_i \mid i \in K\right\}
    \end{align}
    Since $f$ and $g$ were arbitrary functions such that $f(x, y) \neq g(x, y)$, we therefore have:
    \begin{align}
        D^{\text{Tarski}}_{(x, y)} &\leq \sum_{v \in V} D^{\text{Tarski}}_v,
    \end{align}
    Where ``$\leq$" is applied elementwise. By identical reasoning to the $x+y \leq n$ and $x+y \geq n'-n$ cases, we eventually have:
    \begin{align}
        \left\|\Gamma_{\text{Tarski}} \circ D^{\text{Tarski}}_{(x, y)}\right\| &\leq |V| \cdot \max_{i \in [n+1], j \in [n]} \left\|\Gamma_{NOS} \circ D^{NOS}_{i, j}\right\|
    \end{align}
    Which since $|V| \leq 7$ proves \eqref{eq:max-tarski-product-less-than-fifteen-max-nos-product} in this final case, completing the proof.
\end{proof}

\tarskiComplexityMonotone*

\begin{proof}
    We give a black-box reduction. Let $f : \mathcal{L}^{k'}_{n'} \to \mathcal{L}^{k'}_{n'}$ be an instance of $TARSKI(n', k')$. Let $\mathcal{A}$ be an optimal quantum algorithm for $TARSKI(n, k)$. We will construct an algorithm $\mathcal{A}'$ that solves $TARSKI(n', k')$ in the same number of queries.
    
    Let $g : \mathcal{L}^k_n \to \mathcal{L}^{k'}_{n'}$ be the ``clamp function" defined coordinate-wise by:
    \begin{align}
        (g(x))_i &= \min \{x_i, n'\} \qquad \text{for all $i \in [k']$}
    \end{align}
    Then algorithm $\mathcal{A}'$ will run $\mathcal{A}$ on the function $f \circ g$ and return its result. We need to justify three things to complete the proof:
    \begin{itemize}
        \item First, that $\mathcal{A}'$ uses the same number of queries as $\mathcal{A}$. This follows from being able to implement one query to $f \circ g$ with one query to $f$, since $g$ requires no knowledge of $f$ to compute.
        \item Second, that $f \circ g$ is a valid instance of $TARSKI(n, k)$, i.e. that it is monotone. The function $g$ is monotone, since its operation on each coordinate is monotone. Therefore, the composition $f \circ g$ is also monotone.
        \item Third, that any fixed point of $f \circ g$ is also a fixed point of $f$. Let $x$ be a fixed point of $f \circ g$. We must have $x \in \mathcal{L}^{k'}_{n'}$, since $\mathcal{L}^{k'}_{n'}$ is the codomain of $f$. But then $g(x)=x$, so $x = (f \circ g)(x) = f(x)$.
    \end{itemize}
    Therefore $\mathcal{A}'$ uses the same number of queries as $\mathcal{A}$ and is correct at least as often, so it solves $TARSKI(n', k')$.
\end{proof}

\end{document}